\documentclass[12pt]{article}

\usepackage{amsmath,amsthm,amssymb,euscript,epsf,epsfig}
\usepackage{mathtools}
\usepackage{array}
\usepackage{fancybox}
\usepackage{comment}
\usepackage{wasysym}
\usepackage{array}
\usepackage[sort, numbers]{natbib}
\usepackage{hyperref}
\usepackage{stmaryrd} 
\usepackage{graphicx}
\usepackage{relsize}
\usepackage[dvipsnames]{xcolor}
\usepackage{cancel}

\usepackage{tikz,pgf, pbox}
\usepackage{rotating}

\newtheorem{corollary}{Corollary}

\newtheorem{definition}{Definition}
\newtheorem{theorem}{Theorem}
\newtheorem{proposition}{Proposition}
\newtheorem*{proposition*}{Proposition}

\newcommand{\inter}[2]{\llbracket #1;#2 \rrbracket} 
\newcommand{\lrgsize}[1]{\mathlarger{\mathlarger{#1}}}

\newcommand{\beq}{\begin{equation}}
\newcommand{\eeq}{\end{equation}}
\newcommand{\be}{\begin{equation}}
\newcommand{\ee}{\end{equation}}
\newcommand{\bea}{\begin{eqnarray}\displaystyle}
\newcommand{\eea}{\end{eqnarray}}

\newcommand{\C}{{\mathbb C}}

\newcommand{\N}{{\mathbb N}}

\def\g{ \gamma} 
\def\s{ \sigma} 
\def\r{ \rho}

\def\e{{\rm e}}

\def\Sym{ {\rm Sym} } 
\def\Orb{ {\rm Orb}} 
 
\def\Stab{ {\rm Stab} } 
\newcommand{\ov}[1]{\overline{#1}}

\newcommand{\Tr}{{\rm Tr}}

\usepackage[normalem]{ulem}

\begin{document}

\begin{center}

 {\Large \bf  Multiple-Order Tensor Field Theory: \\ Enumeration of unitary invariant observables}

 \medskip

\bigskip

\

{Joseph Ben Geloun$^{(1,3, a)}$,\, Arnauld Solente$^{(1,2,b)}$}

\

 (1) Laboratoire d'Informatique de Paris Nord UMR CNRS 7030 \\ Universit\'e Sorbonne Paris Nord, 99, avenue J.-B. Clement, 93430 Villetaneuse, France 
  
  (2) Magistère de Physique Fondamentale d'Orsay, \\ Université Paris Saclay,
 Rue Louis de Broglie, 91400 Orsay, France. 

  (3) ICMPA-UNESCO Chair, Université d'Abomey-Calavi, 072BP50, Cotonou, Benin  

\ 
  
  $^a$\texttt{bengeloun@lipn.univ-paris13.fr}, 
 $^b$\texttt{arnauldsolente@gmail.com}

  \medskip

  \medskip

\date{}

\begin{abstract}

In Tensor Field Theory (TFT), observables are defined through tensor field contractions that produce unitary invariants for complex-valued tensor fields. Traditionally, these observables are constructed using tensor fields of a fixed order $d$. Here, we propose an extended theoretical framework for TFT that incorporates tensor fields of varying orders $d'$, satisfying $d' \leq d$. We then establish a comprehensive group-theoretic formalism that enables the systematic enumeration of these complex TFT observables. This approach encompasses existing counting methods and therefore recovers known results in specific limiting cases. Additionally, we provide computational tools to facilitate the enumeration of these invariants, unveiling novel integer sequences that have not been documented elsewhere.

\end{abstract}

\end{center}


\newpage

\tableofcontents
 
 \section{Introduction} \label{sec:Intro}

Among the various approaches to quantum gravity (QG), the random geometry paradigm is predicated on defining a set of discrete geometries by randomly assembling simplices and polytopes, each assigned a physically motivated weight. The continuous limit for such geometries is achieved by considering an infinite number of building blocks and simultaneously reducing their area. This approach has yielded significant results in 2D QG through the use of matrix models, which generate, via their Feynman graphs, random discretizations of surfaces \cite{DiFrancesco:1993cyw}. The continuum limit of matrix models is now recognized as the Brownian sphere \cite{LE_GALL_2010}, which is associated with Liouville gravity \cite{DiFrancesco:1993cyw}.

Building on the success of matrix models, tensor models were introduced to address QG in higher dimensions \cite{Ambjorn:1990ge,Sasakura:1990fs,Gross:1991hx}. These models are also quantum field theories, with Feynman graphs corresponding to random discrete geometries in higher dimensions. This perspective has led to numerous advancements. Group Field Theory appears as a lattice gauge field theory, where the gauge group is judiciously chosen to encode essential geometrical data on the Feynman graphs \cite{Boulatov:1992vp,Freidel_2005,Oriti:2006se}. This approach has enabled significant progress toward the emergence of a notion of geometry in the continuum limit. For recent developments in the matter (Ginzburg-Landau theory, Lorentzian Barrett-Crane model), see \cite{Marchetti:2022nrf, Carrozza:2024gnh,carrozza2016tensorial}.

Random tensors and Tensor Field Theory (TFT) \cite{Rivasseau:2011hm,gurau2017random} offer closely related and complementary perspectives on tensor models. In this framework, the models are combinatorial (equilateral geometric data affixed to the discrete geometry) and tailored for statistical and Quantum Field Theory (QFT) computations. The ultimate goal of this quantum field-theoretic approach is to gain insights into the continuum limit of the random geometry attached to its Feynman graphs. As statistical models, the continuum limit of tensor models reveals branched polymer geometries, irrespective of the order $d\ge 3$ of the theory \cite{Bonzom_2011}. This result stands out as both robust and universal, presenting an unexpected and puzzling perspective from a geometric standpoint \cite{Gurau_2013}.

TFT introduces genuine propagating degrees of freedom in random tensors, bridging the gap for the application of advanced QFT methods to tackle the same problem. Perturbative analyses have demonstrated the renormalizability of numerous TFTs \cite{BenGeloun:2011rc}, unveiling remarkable properties such as asymptotic freedom and asymptotic safety within specific regimes \cite{BenGeloun:2013vwi, carrozza2016tensorial}. At the non-perturbative level, inspired by preliminary studies of matrix models \cite{Eichhorn:1309,Eichhorn:1408}, the Functional Renormalization Group has been systematically applied to TFT across various dimensions \cite{Benedetti:1411, BenGeloun:1508, BenGeloun:1601, Krajewski:2016jf, Eichhorn:1811, Eichhorn:2018ylk, Eichhorn:2019hsa, Eichhorn:2020gq, Lahoche:2019wm, Lahoche:2020bi, Lahoche:2020df, Lahoche:2020ib}. These efforts have yielded compelling evidence suggesting the existence of fixed points in the large $N$ limit and for non-compact configuration space of the fields. This set of results has opened promising avenues of research, framing QG as random geometry endowed with a continuum limit tool.

On another front, efforts to move beyond the phase of branched polymer geometry have led to the introduction of new tensor models and TFTs, known as enhanced models \cite{Bonzom:2015axa, Ben_Geloun_2018}. Under this hypothesis, partial progress in overcoming branched polymer geometry has been achieved by properly rescaling suppressed configurations in the large $N$ limit, thereby allowing them to contribute at criticality. This area of research is still evolving, with additional phases being discovered for tensor models, such as the Brownian map phase and phases involving multiple 2D universes \cite{Bonzom:2015axa}. Furthermore, TFT models with enhanced interactions demonstrate that non-melonic diagrams are (super-)renormalizable \cite{Ben_Geloun_2018}. These diagrams significantly influence the RG flow, modifying their asymptotic behavior—departing from asymptotically free to asymptotically safe, or, in certain models, transitioning to regimes that defy both classifications \cite{Geloun:2023oyd}. Nonetheless, finding a significant geometrical phase in dimensions strictly greater than 2 at the continuum limit in tensor models or within TFTs remains a challenge. A crucial question to address is how to explore the Tensor Theory Space and construct TFT actions with meaningful interactions. This paper delves into this issue and proposes approaches to tackle it.

The general approach to writing an interaction for a TFT model involves tensor contractions for a tensor field of fixed order $d$ (the interaction then represents a polytope of dimension $d$). Other working hypotheses could be explored; specifically, one might allow tensor contractions of different orders. The interaction would then be associated with the gluing of polytopes of different dimensions. This hypothesis would enrich the landscape of TFT models that describe discrete spaces of dimension $d'\le d$. Finally, these new models could lead to the discovery of new continuous limits.

The primary aim of this paper is to address a foundational question: given $n_1$ tensors of order 1, $n_2$ tensors of order 2, $n_3$ tensors of order 3, and so on, how many non-equivalent contractions can be constructed? Physically, this translates to determining the total number of possible interactions between these tensor fields. It is noteworthy that the focus is not on the specific nature of the interactions but solely on their enumeration. Crucially, each tensor interaction corresponds uniquely to an orbit of a particular group action on a given set, and identifying this orbit characterizes the interaction itself. Additionally, each interaction can be visually represented as a graph. This renders the enumeration of interactions equivalent to counting all possible graphs with specific constraints.

\par

This investigation builds upon a solid foundation within the field. Seminal work by Read \cite{Read1959TheEO} established a framework for counting locally constrained graphs—those defined by fixed valence and edge coloring—through the application of group theory. Recently, de Mello Koch and Ramgoolam developed this method and exported it to the realm of QFT, where they enumerated various types of Feynman graphs arising from matrix models, String Theory, and Quantum Electrodynamics \cite{deMelloKoch:2011uq}. In the context of TFT, the enumeration of interactions and observables for complex tensor field models has been systematically addressed in \cite{BenGeloun:2013lim}. This work has provided a detailed account of unitary invariants of fixed order $d$ for $n$ tensor fields $T$ and $n$ tensor fields $\bar{T}$. Real tensor field models have also been explored in \cite{Avohou_2020}, with an emphasis on counting orthogonal invariants. Building on this result, Avohou et al. recently extended these findings to tensors that permit both unitary and orthogonal transformations in \cite{Avohou:2024agh}. Each of these studies deploys group theory to count graphs with fixed valence, corresponding to interactions within tensor models. Our approach aligns with and generalizes these previous efforts in the unitary case: we leverage group theory to enumerate the interactions and observables of complex TFTs. Specifically, this involves the enumeration of unitary invariants arising from tensor contractions of various orders, each bounded by a given integer $d$. 
An important remaining question concerns the potential range of applications for the enumeration and the methods developed in this paper.

The enumeration of tensor observables has provided valuable insights across several domains \cite{BenGeloun:2013lim,Avohou_2020, Avohou:2024agh,BenGeloun:2017vwn,BenGeloun:2020yau,BenGeloun:2020lfe,BenGeloun:2021cuj}, particularly as it exhibits a strong correspondence with Topological Field Theory. Remarkably, each of the aforementioned enumerative results aligns with the counting of branched covers of various topological structures. This connection has been formalized using Dijkgraaf-Witten theories \cite{DW}, employing the symmetric group as gauge group and a lattice encoding the gauge orbit corresponding to a tensor invariant. For instance, the enumeration of unitary invariants of order 3, constructed from $2n$ tensors, corresponds to the counting of branched covers of the sphere with three punctures \cite{BenGeloun:2013lim}. Building on this foundational result, numerous other enumerative approaches have been developed in direct correspondence with tensor model observables \cite{Avohou_2020, Avohou:2024agh}. These findings further demonstrate the intimate link between the combinatorics of tensors and Topological Field Theory.

Another key observation is that performing a representation/Fourier-theoretic transform of any of the previous enumeration formulas introduces insights from the realm of Combinatorial Representation Theory \cite{BarceloRam1997}. As highlighted in a series of works \cite{BenGeloun:2017vwn,BenGeloun:2020yau,BenGeloun:2020lfe,BenGeloun:2021cuj}, the enumeration of unitary tensor observables has addressed the longstanding question posed by Murnaghan regarding the combinatorial interpretation of the Kronecker coefficient \cite{Murnaghan1938TheAO}. Specifically, the Kronecker coefficient enumerates certain vectors within operator kernels, which act on an algebra whose basis is labeled by unitary tensor invariants. This result was unveiled through the framework of enumerating tensor invariants. Consequently, it becomes crucial to explore broader and more general questions regarding the enumeration of tensor observables within a wider context, as undertaken in this work.
It is easy to advocate that the interplay between Combinatorics, Representation Theory, and Topological Field Theory has consistently been a central theme in Theoretical Physics. Recently, this dynamic has been revitalized by a wave of transformative results that uncover connections at the intersection of String Theory, Quantum Information, and (quantum) Complexity Theory \cite{Padellaro:2025jtd,dMelloKoch:2021,Ramgoolam:2022xfk,Padellaro:2023oqj}. The present paper is undoubtedly seminal and opens another playground for exploring possible connections between these very same topics.

The structure of this paper is as follows: in Section \ref{sec:Revue}, we take a fresh look at complex-valued tensor contractions, exploring their definition as unitary invariants and showing how they can be grouped into equivalence classes. To make these ideas more concrete, we include several examples along the way. This section is largely based on the work drawn from \cite{BenGeloun:2013lim}, and focuses in particular on how such interactions can be enumerated for tensors of a fixed order.
Section \ref{sec:ordremult}, which introduces the novel contribution of this paper, extends the formalism to encompass tensors of all orders through the definition of multiple-order tensor contractions. In this broader framework, we revisit the property of unitary invariance. The main results of this paper are:

- the definition of multiple-order tensor contractions, Definition \ref{def:CtrOrdreMult}, 

- the proof that  they are unitary invariants, Theorem \ref{theo:InvarianceUnderU}, 

- the enumeration formula for multiple-order tensor interactions,  Theorem  \ref{theo:CountMultiOrder}, 
 that can be directly implemented, see Appendix \ref{app:Code}. 
 
 In section \ref{sec:application}, we provide a comprehensive application of our main theorem, alongside validations through a combinatorial case.
Section \ref{sec:concl} is our conclusion which summarizes our findings and outlines a selection of problems for future exploration. The manuscript concludes with multiple appendices. Appendix \ref{app:Inv} presents the proof of the unitary invariance of complex tensor contractions at fixed order and their invariance under a certain group action. Appendix \ref{sec:CalculOrbites} includes numerical data on the enumeration for a finite set of tensors. Finally, Appendix \ref{app:Code} offers an extensive collection of the codes developed to illustrate the counting formulae 
of the paper.

\section{Counting unitary tensor invariants: a review}
\label{sec:Revue}

In this section, we examine the construction and enumeration of interactions in complex tensor models at a fixed order $d$, with the goal of generalizing the formalism in the subsequent section \ref{sec:ordremult}.
Thus, the content of this section is familiar territory \cite{BenGeloun:2013lim}, but presenting it this way simplifies and supports the proof of the following statements.

\subsection{Tensor contractions and unitary group invariance}\label{subsec:DefContraction}

We provide a definition of a tensor (e.g., see Lang's book \cite{Lang2002} chap. XVI, p554). Note that this is not the most general one, but it is largely sufficient for our purpose.
It is also noted that we use the summation convention that repeated indices are implicitly summed. 

\begin{definition}(Tensor) \label{def:tenseur}
Let $d \in \N^*$ and $E$ be a $\C$-vector space of dimension $N$. A tensor of order $d$, $T:E^{\times d} \rightarrow \C$  is a multilinear form. Given $B_1 = \{e_1, e_2,\cdots, e_N\}$ a basis of $E$, we define the coefficients of $T$ by: 
\bea 
\forall i_1, \cdots, i_d \in \inter{1}{N}, \quad T_{i_1\cdots i_d} = T(e_{i_1}, \cdots, e_{i_d}). 
\eea
We say that $i_1$ is the index of color 1, $i_2$ the index of color 2, etc. The coefficients of the conjugate tensor $\ov{T}$ are defined by:
\bea 
\forall i_1, \cdots, i_d \in \inter{1}{N}, \quad \ov{T}_{i_1\cdots i_d} = \ov{T(e_{i_1}, \cdots, e_{i_d})}. 
\eea
\end{definition}  

Given another basis $B_2 = \{f_1,f_2, \cdots, f_N\}$ of $E$, and denoting 
$\Lambda \in GL(N, \C)$, the change of basis matrix from  $B_1$ to $B_2$, $f_i = \Lambda_{ij}e_j$, then the coefficients of $T'_{i_1\cdots i_d}$ of $T$ and $\ov{T}'_{i_1\cdots i_d}$ of $\ov{T}$ in basis $B_2$ are given by: 
\bea 
\label{eq:chtTenseurBase} 
T'_{i_1\cdots i_d} = \Lambda_{i_1 j_1}\cdots\Lambda_{i_d j_d}T_{j_1\cdots j_d} \quad \text{ and } \quad\ov{T}'_{i_1\cdots i_d} = \ov{\Lambda}_{i_1 j_1}\cdots\ov{\Lambda}_{i_d j_d}\ov{T}_{j_1\cdots j_d} \;. 
\eea

\begin{definition}[Tensor contraction] \label{def:ContrcUnitaireMonoTenseur} Let $d,n \in \N^*$, $T$ be a tensor of order $d$ and $\ov{T}$ be its conjugate. A contraction of $n$ tensors $T$ of order $d$ and $n$ tensors $\ov{T}$ with same order is the sum over the indices of the tensors such that an index of color $c$ of a tensor $T$ is summed with another index of color $c$ of a tensor $\ov{T}$. To indicate how the indices are contracted, we  use a $d$-tuple of permutations $\s = (\s_1, \s_2, \cdots, \s_d) \in S^{\times d}_n$, where $S_n$ is known as the symmetric group of $n$ objects. Thus a tensor contraction is denoted $I(\s;T)$ and is written as:
\bea  I(\s;T) = K(\s, \{i_1^{(j)}, \cdots, i_d^{(j)}\}_j, \{\bar i_1^{(k)} ,\cdots, \bar i_d^{(k)}\}_k) \Big[\prod_{j = 1}^n T_{i_1^{(j)} \cdots i_d^{(j)}} \Big]\Big[ \prod_{k=1}^n \ov{T}_{\bar i_1^{(k)} \cdots \bar i_d^{(k)}}\Big]~, 
\eea 
where $K$, called the contraction kernel, is given by
\bea 
K(\s, \{i_1^{(j)}, \cdots, i_d^{(j)}\}_j, \{\bar i_1^{(k)} ,\cdots,\bar i_d^{(k)}\}_k)  = \prod_{l=1}^n \prod_{c=1}^d \delta(i_c^{(l)}, \bar i_c^{(\s_c(l))}) ~,
\eea
where the Kronecker symbol $\delta(i,j) = 1$ if $i=j$, and 0 otherwise.
\end{definition} 

We have provided two explicit examples of definition \ref{def:ContrcUnitaireMonoTenseur} in subsection \ref{subsec:ExamplesAndGraphLink}. A tensor contraction admits a graphical representation that simplifies its comprehension, which we will shortly explore. Nonetheless, the above mathematical formula provides a straightforward algebraic tool for demonstrating the following properties. The next proposition asserts that tensor contractions exhibit two kinds of invariances: they are invariant under the action of the unitary group and under a specific permutation group action, which requires clarification. The proof of this statement is given in appendix \ref{app:Inv}.

\begin{proposition}\label{prop:InvarianceContraction}
Let $d,n \in \N^*$, $n$ tensors $T$ and $n$ tensors $\ov{T}$, each of order $d$ and each defined over a complex vector space $E$ of dimension $N$. Denote by $U(N)$ the unitary group, and consider the natural action of $U(N)^{\otimes d}$ on tensors of order $d$ via the fundamental representation on each index.

Let $\s \in S_n^{\times d}$ and
 let $I(\s; T)$ denote a specific index contraction between the tensors $T$ and $\ov{T}$ determined by $\s$.
 Then, 
\begin{itemize}
\item[(i)] $I(\s; T)$ is invariant under the action of $U(N)^{\otimes d}$ on the tensors $T$ and $\ov{T}$.
\item[(ii)]
 $\forall \g, \r \in S_n$, $I((\r\s_1\g, \cdots, \r\s_d\g); T) = I((\s_1, \cdots, \s_d); T) = I(\s; T)$.
\end{itemize}
\end{proposition}

We denote the action $(ii)$ more simply as  $I(\r\s\g;T) = I(\s;T)$.

\subsection{Unitary invariants as colored bipartite graphs} \label{subsec:ExamplesAndGraphLink}
In this section, we examine some special cases of tensor contractions and introduce a graphical method to represent them. Since our subsequent developments focus on tensor contractions of various orders, it makes sense to review the representation corresponding to each order incrementally.

A tensor of order 1 is a vector denoted by $\phi=(\phi_i)$. 
We illustrate a unitary contraction for $n=3$ vectors $\phi$ and 3 conjugate vectors $\ov{\phi}$ with the permutation $\s_1=(1)(2)(3)$: 
\bea \label{eq:6vectorContraction}
\begin{aligned}
I(\s_1; \phi) &= \overbrace{ \delta(i_1^{(1)}, \bar i_1^{\s_1(1)})\delta(i_1^{(2)}, \bar i_1^{\s_1(2)})\delta(i_1^{(3)}, \bar i_1^{\s_1(3)}) }^{ K(\s_1, \{i_1^{(j)}\}_{j}, \{\bar i_1^{(k)}\}_{k}), \,\, \, j=1,2,3 \text{ and } k =1,2,3} \phi_{i_1^{1}}\phi_{i_1^{2}}\phi_{i_1^{3}}\ov{\phi}_{\bar i_1^{1}}\ov{\phi}_{\bar i_1^{2}}\ov{\phi}_{\bar i_1^{3}} \\ 
&= \delta(i_1^{(1)}, \bar i_1^{(1)})\delta(i_1^{(2)}, \bar i_1^{(2)})\delta(i_1^{(3)}, \bar i_1^{(3)}) \phi_{i_1^{1}}\phi_{i_1^{2}}\phi_{i_1^{3}}\ov{\phi}_{\bar i_1^{1}}\ov{\phi}_{\bar i_1^{2}}\ov{\phi}_{\bar i_1^{3}} \\ 
&= \phi_{a}\ov{\phi_a}\phi_{b}\ov{\phi_b}\phi_{c}\ov{\phi_c}\\
&= ||\phi||^2||\phi||^2||\phi||^2
\end{aligned}
\eea
where we easily target the kernel $K$  of the contraction. 
Checking that this contraction is stable under $U(N)^{\otimes 3}$ is trivial since the norm is a unitary invariant. By changing the permutation $\s_1$ for any other element of  $S_3$, we get the same contraction.

We can associate a graph with this contraction. Each vector $\phi$  is associated with a vertex, say white, with a half-edge, and each vector $\ov{\phi}$ is associated with a vertex, say black, with a half-edge. When a vector $\phi$ and $\ov{\phi}$ have the same contraction index, the two half-edges are connected. An illustration is given in figure \ref{fig:Cntrc3Vec}.
\begin{figure}[ht]
\center
\includegraphics[scale=0.20]{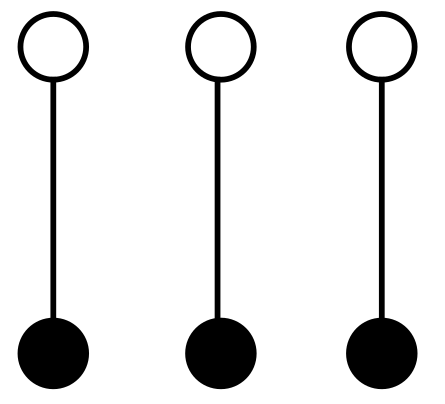}
\caption{Graph associated with the vector contraction $I(\s_1; \phi)$,
$\s_1=(1)(2)(3)$.}
\label{fig:Cntrc3Vec}
\end{figure}
Graphs representing contractions of $2n$ vectors, as their norm, simplify to bipartite segments.

Let us examine tensors of order 2, specifically matrices. Take $M$ as a matrix with components $(M_{ij})$. For a unitary contraction involving $n = 2$ matrices $M$ and $2$ matrices $\ov{M}$, we focus on the 2-tuple of permutations
\bea \label{eq:SforM2}
\s= (\s_1=(1)(2), \s_2=(12)) 
\eea 
\bea \label{eq:4matricxContraction}
\begin{aligned}
I(\s; M) =& \overbrace{\delta(i_1^{(1)}, \bar i_1^{\s_1(1)})\delta(i_2^{(1)}, \bar i_2^{\s_2(1)})\delta(i_1^{(2)}, \bar i_1^{\s_1(2)})\delta(i_2^{(2)}, \bar i_2^{\s_2(2)})}^{K(\s_1, \{i_1^{(j)}, i_2^{(k)}\}_{j}, \{\bar i_1^{(k)}, \bar i_2^{(k)}\}_{k}), \,\, \, j=1,2 \text{ and } k =1,2} \\
&M_{i_1^{(1)}i_2^{(1)}}M_{i_1^{(2)}i_2^{(2)}}\ov{M}_{\bar i_1^{(1)} \bar i_2^{(1)}}\ov{M}_{\bar i_1^{(2)} \bar i_2^{(2)}}\\
=& \delta(i_1^{(1)}, \bar i_1^{(1)})\delta(i_2^{(1)}, \bar i_2^{(2)})\delta(i_1^{(2)}, \bar i_1^{(2)})\delta(i_2^{(2)}, \bar i_2^{(1)}) \\
&M_{i_1^{(1)}i_2^{(1)}}M_{i_1^{(2)}i_2^{(2)}}\ov{M}_{\bar i_1^{(1)} \bar i_2^{(1)}}\ov{M}_{\bar i_1^{(2)} \bar i_2^{(2)}}\\
=&M_{ab}M_{cd}\ov{M}_{ad}\ov{M}_{cb} \\
=& \Tr([MM^{\dag}]^2)
\end{aligned}
\eea
As with vectors, one simply proves that the contractions by traces of products of matrices $MM^{\dag}$ is stable under unitary group. 

We represent a matrix $M$ (resp. $\ov{M}$)) by a vertex, for example, white (resp. black), with two half-edges of different colors. If a matrix $M$ and $\ov{M}$ have the same contraction index on the same color, then we connect the two half-edges of the same color. An illustration of $I(\s; M)$ with 
$\s$ given by \eqref{eq:SforM2} appears  in figure \ref{fig:Cntrc2Mat}.
\begin{figure}[ht]
\center
\includegraphics[scale=0.50]{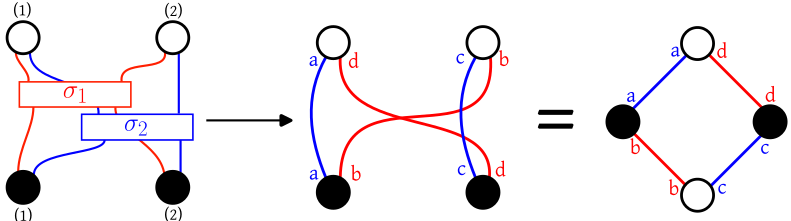}
\caption{Graph associated with the matrix contraction $I(\s; M)$, 
with $\s=(\s_1=(1)(2), \s_2=(12))$.}
\label{fig:Cntrc2Mat}
\end{figure}
We conclude that matrix contractions of trace forms yield
bipartite cycles and a collection of those. 

The generalization of the graphical representation of tensor contractions of order $d$ is done similarly: consider $n$ tensors $T$ and $n$ tensors $\ov{T}$ of order $d$ that are contracted by the $d$-tuple $\s = (\s_1,\cdots, \s_d)$. A tensor $T$ (resp. $\ov{T}$) is represented by a vertex, say white (resp. black), with $d$ half-edges of different colors. If a tensor $T$ and a tensor $\ov{T}$  share a contraction index of the same color $i_c^{(j)}$, then the half-edge of color $c$ of $T$ is connected to the half-edge of color $c$ of $\ov{T}$. Examples are given in figure \ref{fig:GrphVarie}. This construction rule produces bipartite 
colored graphs 
with fixed valence determined by the order $d$ of the tensor $T$.

It is important to highlight a key aspect. The pair of graphs on the left of figure \ref{fig:Cntrc2Mat} provides a graphical interpretation of the gauge group invariance described in point $(ii)$ of proposition \ref{prop:InvarianceContraction}. To formally represent a contraction, it is necessary to use permutations and assign numbers to the vertices. The $d$-tuple of permutations acts on these indices, connecting half-edge to half-edge. However, this numbering process creates distinguishable vertices from those that were initially indistinguishable—a concept related to the bosonic symmetry of $n$ identical tensors $T$ and $n$ identical tensors $\ov{T}$.
To express the contractions formally, labels are introduced, but these labels are ultimately removed by grouping the contractions into equivalence classes. At the graphical level, gauge invariance simply corresponds to invariance under the relabeling of vertices.

\begin{figure}[ht]
\center
\includegraphics[scale=0.23]{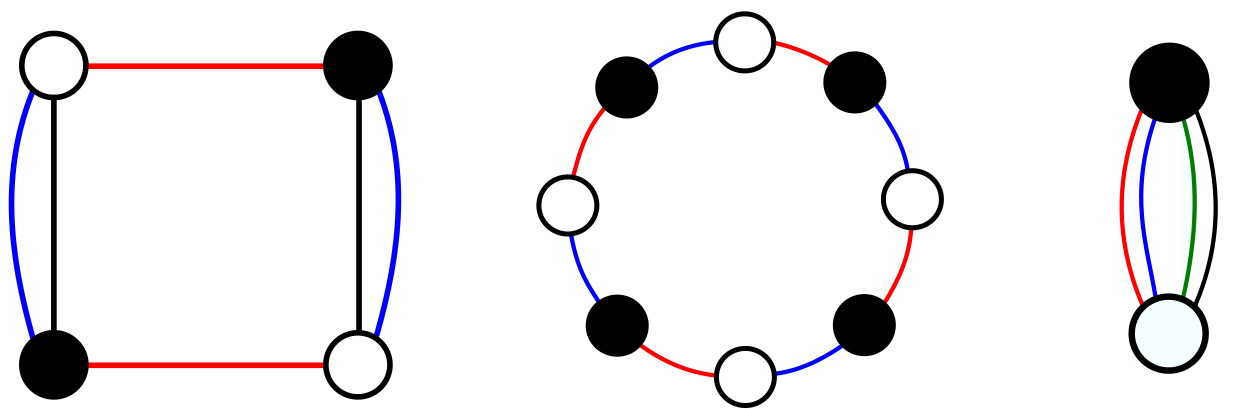}
\caption{Illustration of $T_{\color{blue}a\color{black}b\color{red}c}T_{\color{blue}d\color{black}e\color{red}f}\ov{T}_{\color{blue}a\color{black}b\color{red}f}\ov{T}_{\color{blue}d\color{black}e\color{red}c}$ (left), $M_{\color{red}a\color{blue}b}M_{\color{red}c\color{blue}d}M_{\color{red}e\color{blue}f}\ov{M}_{\color{red}a\color{blue}d}\ov{M}_{\color{red}e\color{blue}b}\ov{M}_{\color{red}c\color{blue}f}$ (middle) and $T_{\color{red}a\color{blue}b\color{OliveGreen}c\color{black}d}\ov{T}_{\color{red}a\color{blue}b\color{OliveGreen}c\color{black}d}$ (right).}
\label{fig:GrphVarie}
\end{figure}

\subsection{Counting invariants}\label{sub:comptageInvMonoOrder}
The main tool for counting is Burnside's lemma and the orbit-stabilizer theorem.
We recall these pivotal propositions in appendix \ref{app:Inv}.

Below, we define the gauge group action that describes the invariance of unitary contractions according to point $(ii)$ of proposition \ref{prop:InvarianceContraction}.

\begin{definition}[Group action]\label{def:GroupActionMonoOrder}

Let $d,n \in \mathbb{N}^*$ and $H_1 = H_2=S_n$. We define the group action $\alpha^d_n$ on $(S_n)^{\times d}$ by:
\bea \alpha^d_n: \left\{ \begin{array}{ll}
H_1 \times H_2 \times (S_n)^{\times d} \rightarrow (S_n)^{\times d}\\
(\gamma, \rho, \sigma) \mapsto \gamma \sigma \rho = (\gamma\sigma_1\rho,\gamma\sigma_2\rho, \cdots, \gamma\sigma_d\rho)
\end{array} \right. \eea

\end{definition}

Proposition \ref{prop:InvarianceContraction} shows that any contraction $I(\s; T)$ is invariant under this group action.

\subsubsection{Counting formula}
\begin{theorem}[Number of orbits of unitary invariants of order $d$ \cite{BenGeloun:2013lim}]\label{theo:JBandSajaye}
Let $d, n \in \mathbb{N}^*$. We denote by $Z^d_n$ the number of orbits of the action $\alpha^d_n$. We have 
\bea 
Z^d_n = \sum_{p ~ \vdash ~ n} \Sym(p)^{d-2}. \eea
\end{theorem}
\begin{proof} Let $G = S_n^{\times d}$, recall that $H_1=H_2=S_n$ and define $\delta: G \rightarrow G$  which is 1 on the identity and 0 everywhere else. Starting from Burnside's lemma (see proposition \ref{prop:Burnside}) for the action $\alpha^d_n$, one obtains
\bea
\begin{aligned}
Z^d_n &= \frac{1}{|H_1 \times H_2|} \sum_{(\gamma, \rho) \in H_1 \times H_2} |\left\{\sigma=(\sigma_1, \cdots, \sigma_d) \in G ~|~ \gamma\sigma \rho = \sigma \right\}|\\
&=  \frac{1}{|H_1|~|H_2|} \sum_{(\gamma, \rho) \in H_1 \times H_2} ~\sum_{\sigma_1, \cdots, \sigma_d \in S_n} \prod_{i = 1}^d \delta(\gamma\sigma_i\rho\sigma_i^{-1}).
\end{aligned}
\eea
Noting that $\delta(\gamma\sigma_i\rho\sigma_i^{-1}) = 1$
entails that $\gamma$ and $\rho^{-1}$ are conjugate permutations. Since $\rho$ and $\rho^{-1}$ are necessarily conjugate ($S_n$ is an ambivalent group), 
we infer that $\gamma$ and $\rho$ are conjugate. 
Therefore, we add this information by saying that
\bea \delta_c(\gamma, \rho)\delta(\gamma\sigma_i\rho\sigma^{-1}) = \delta(\gamma\sigma_i\rho\sigma^{-1})~, \eea
by defining $\delta_c: H_1 \times H_2 \rightarrow \{0, 1\}$ which is 1 if $\gamma$ and $\rho$  are conjugate, 0 otherwise. It becomes interesting to change the summation order according to the conjugacy classes. Since $S_n$ 
gets partionned in conjugacy classes $C(\mu)$,
where $C(\mu)$ is the conjugacy class associated with the partition $\mu$ of $n$, we have
\bea \sum_{(\gamma, \rho) \in H_1 \times H_2} = \sum_{\gamma \in H_1} \sum_{\rho \in H_2} = \sum_{\mu ~ \vdash ~ n} \sum_{\nu ~ \vdash ~ n} \sum_{\gamma \in C(\mu)} \sum_{\rho \in C(\nu)}~,\eea
and 
\bea
\begin{aligned}
Z^d_n &=  \frac{1}{|H_1|~|H_2|} \sum_{\mu ~ \vdash ~ n} \sum_{\gamma \in C(\mu)} \sum_{\rho \in C(\mu)} ~ \prod_{i = 1}^d \left( \sum_{\sigma_i \in S_n}\delta(\gamma\sigma_i\rho\sigma_i^{-1}) \right).
\end{aligned}
\eea
For a conjugate pair $\gamma$ and $\rho^{-1}$, there exists $\eta \in S_n$ such that $\rho^{-1} = \eta^{-1} \gamma \eta$ 
and 
\bea \sum_{\sigma \in S_n} \delta(\gamma\sigma\rho\sigma^{-1}) = \sum_{\sigma \in S_n} \delta(\gamma(\sigma\eta)\gamma^{-1}(\sigma\eta)^{-1}) = \sum_{\sigma \in S_n} \delta(\gamma\sigma\gamma^{-1}\sigma^{-1})~, \eea
where, for the last line, we used 
$\delta(\gamma\sigma\gamma^{-1}\sigma^{-1}) = 1$. This also means that $\sigma \in \Stab(\gamma)$, where $\Stab(\gamma)$ stands for 
 the stabilizer of $\g$ associated with the conjugate action. 
 This boils down to: 
\bea 
\sum_{\sigma \in S_n} \delta(\gamma\sigma\gamma^{-1}\sigma^{-1}) = |\Stab(\gamma)| \equiv \Sym(\gamma).
\eea
Taking this into account, we re-express the main counting formula as
\bea 
Z^d_n = \frac{1}{|H_1|~|H_2|} \sum_{\mu ~ \vdash ~ n} \sum_{\gamma \in C(\mu)} \sum_{\rho \in C(\mu)} ~ \Sym(\gamma)^d. 
\eea
Since the symmetry factor of a permutation depends only on its conjugacy class:
\bea
\begin{aligned}
Z^d_n &= \frac{1}{|H_1|~|H_2|} \sum_{\mu ~ \vdash ~ n} \sum_{\gamma \in C(\mu)} \sum_{\rho \in C(\mu)} ~ \Sym(C(\mu))^d\\
&= \frac{1}{|H_1|~|H_2|} \sum_{\mu ~ \vdash ~ n} ~ |C(\mu)|^2\,  \Sym(C(\mu))^d.
\end{aligned}
\eea
Now apply the orbit-stabilizer theorem (see proposition \ref{prop:orbstab}) for the inner automorphism action
to obtain 
\bea |C(\mu)|\Sym(\mu) \equiv |C(\mu)|\Sym(C(\mu)) = |S_n|. \eea
Finally, dealing with diagonal actions with the same groups $H_1 = H_2 = S_n$, we obtain
\bea Z^d_n = \frac{1}{|H_1|~|H_2|} \sum_{\mu ~ \vdash ~ n} ~ \frac{|S_n|^2}{\Sym(\mu)^2} \Sym(\mu)^d = \sum_{\mu ~ \vdash ~ n} ~ \Sym(\mu)^{d-2}. \eea
\end{proof}

\subsubsection{Comments on low orders}
For $d=1$, the formula gives: 
\bea Z^1_n = \sum_{\mu ~ \vdash ~n} \frac{1}{\Sym(\mu)} = \frac{1}{|S_n|}\sum_{\mu ~ \vdash ~n} |C(\mu)| = 1
\eea 
confirming that there  is only one vector invariant for all $n\ge 1$. 
At $d=2$, one gets $Z^2_n = \sum_{\mu ~ \vdash ~n}  1 =p(n) $ the number of partitions of $n$. This coincides with the number of ways
of splittings the product of $n$ matrices $M$  and $n$ matrices $M^\dag$
to make distinct trace products.  

\section{Counting unitary contractions of multiple-orders}
\label{sec:ordremult}

This section undertakes the analysis of multiple-order tensor invariants. We begin with a series of definitions that lay the foundation for the concept of multiple-order tensor contraction and its graphical representation. Using this formalism, we derive a counting formula for these new invariants.

\subsection{Encoding a bipartite colored graph}\label{subsec:EncodingGraph}

In the rest of the section, $\forall k \in  \N^*$, $[k]$ denotes the integer interval $\inter{1}{k}$. Consistently, $P([k])$ is  power set of $\inter{1}{k}$.
Refereing to the order $d$ of a tensor, 
an element of $[d]$ is called color.

\begin{definition}[Colored vertex and type]\label{def:coloredVertex} 
Let $d \in \N^*$. Let $V$ be a set, the elements of which are called vertices. A colored vertex is a pair $(v,A)$, where $v\in V$ 
and where $A \in P([d])$. The set $A$ is called the  color type of the vertex $v$, 
and its cardinality $|A|$ indicates the number of half-edges attached to the vertex $v$ (\textit{i.e.} the valence of $v$). 
The elements of $A$ are called colors, each color is associated with a unique half-edge of $v$. 

We call a color type function $\phi$ over $V$ a function $\phi: V \rightarrow P([d])$, which satisfies 
$\phi^{-1}(\varnothing) = \varnothing$.  
More simply, $\phi$ is called a type. A colored set of vertices with $d$ colors is a set of vertices equipped with a type $\phi:V \rightarrow P([d])$ noted down as:
\bea
\Lambda = (V, \phi)\,. 
\eea
\end{definition}

Given a colored set of vertices with $d$ colors, $\Lambda = (V, \phi)$, we assign to each vertex $v \in V$ the colored vertex $(v, \phi(v))$.  
We also note that if we use a type  $\phi: V \to P([d])$,  
then for any $d' \geq d$, one can define another type  
$\phi': V \to P([d'])$ that assigns the same colors to the vertices as $\phi$, \textit{\i.e}., for any $v\in V$, $\phi'(v) = \phi(v)$.  
Therefore, $d$ should be chosen as the minimal number of colors 
for which the definition is meaningful.

\begin{definition}[Colored section, color multiplicity and chromatic index]
\label{def:labelingColorMult} Let $V$ be a set of vertices and $\Lambda = (V, \phi)$ be a colored set of vertices with 
$d$ colors. 
A colored section $V|c$, with $c \in [d]$, is the set of all colored vertices of $V$ having the color $c$:
\bea 
V|c = \{ v \in V\, | \, c \in \phi(v) \}.
\eea
The color multiplicity of the color $c \in [d]$ for $\Lambda$ is the cardinality of $V|c$, \textit{i.e.}:
\bea 
\forall c \in [d],\,  m_c = |\, V|c \,|. 
\eea
$d$ is said to be the chromatic index if and only if 
\bea \forall c \in [d],\, m_c > 0 .\eea
\end{definition}

One realizes that if $d$ is the chromatic index, 
then each color appears at least once in the color type 
of some vertex. This means that all colors 
in $[d]$ are used.

\begin{definition}[Compatible colored sets of vertices] Let $\Lambda= (V, \phi)$, $\Lambda'=(V', \phi')$ be two colored sets of vertices of chromatic index $d \in \N^*$. Then, $\Lambda$ and $\Lambda'$ are said to be compatible if they have disjoint sets of vertices and share the same multiplicity for each color, in other words: 
\bea
&& (i) \quad \, V \cap V' = \varnothing, \crcr
&& (ii)\quad \, \forall c \in [d],\, m_c = m'_c \,. 
\eea
\end{definition}

\begin{definition}[Colored bijection and bipartite colored graph]\label{def:bipartiteColoredGraph}
Let $\Lambda_1 = (V_1, \phi_1)$, $\Lambda_2=(V_2, \phi_2)$ be two compatible colored sets of vertices of chromatic index $d \in \N^*$. For any color $c \in [d]$, let $S(c)$ denote the set of bijections from $V_1 | c$ to $V_2|c$. A colored bipartite graph $G$ of chromatic index $d$ is the data of $\Lambda_1$, $\Lambda_2$ and an element $\s$ of $S = \bigtimes_{c= 1}^d S(c)$. Explicitly, 
\bea G = (\Lambda_1, \Lambda_2, \s).\eea 
\end{definition}
    
The above definition coincides with the standard notion of an edge-colored bipartite graph $G$, whose vertex set is the disjoint union of two sets $V_1$ and $V_2$. It remains to specify the set of edges in this context, along with the associated coloring.

\begin{proposition} \label{prop:edges}Let $G= (\Lambda_1, \Lambda_2, \s)$ be a colored bipartite graph of chromatic index $d \in \N^*$ with $\Lambda_1 = (V_1, \phi_1)$, $\Lambda_2 = (V_2, \phi_2)$ and $\s = (\s_1, \dots, \s_d) \in  S = \bigtimes_{c= 1}^d S(c)$. Let $E$ be the set of pairs defined
by 
\bea 
\label{eq:edgesSet} 
E = \left\{ (v, \s_c(v)) \, | \,  c \in [d] \text{ and } v \in V_1|c\right\}.
\eea
Then, 
\bea 
&& (i) \quad E = \bigcup_{c = 1}^d \left\{ (v, \s_c(v)) \, | \,  v \in V_1|c\right\} = \bigcup_{v \in V_1} \left\{ (v, \s_c(v)) \, | \, c \in \phi_1(v)\right\} \crcr
&& (ii) \quad E = \bigcup_{c = 1}^d \left\{ (\s_c^{-1}(v), v) \, | \,  v \in V_2|c\right\} = \bigcup_{v \in V_2} \left\{ (\s^{-1}_c(v), v) \, | \, c \in \phi_2(v)\right\} 
\eea
\end{proposition}
\begin{proof}
One easily proves the statement, keeping in mind that,  
\bea  
\forall c \in [d],\, \forall v \in V_1, \,\,\,\, v \in V_1 |c \Leftrightarrow c \in \phi_1(v)\,.  
\eea
Claim $(i)$ follows by unpacking the definition \ref{eq:edgesSet}, and the fact that the set of colors partitions the edge set. To establish the equivalence between $(i)$ and $(ii)$, one uses the bijection between $V_1|c$ and $V_2|c$.
\end{proof}

\begin{proposition} \label{prop:linkWithCOnventionalGraph}
Let $G= (\Lambda_1, \Lambda_2, \s)$ be a colored bipartite graph of chromatic index $d \in \N^*$, with $\Lambda_1 = (V_1, \phi_1)$, $\Lambda_2 = (V_2, \phi_2)$ and $\s = (\s_1, \dots, \s_d) \in S = \bigtimes_{c= 1}^d S(c)$.  Then,  the graph $(V,E)$ formed from the set of vertices $V = V_1 \cup V_2$ and the set of edges given by $E=\bigcup_{c=1}^d \left\{ (v,\s_c(v)) \, | \, v \in V_1|c \right\}$ is an 
edge-colored bipartite graph of chromatic index $d$,
in the conventional sense. 
\end{proposition}
\begin{proof}
By the definition of compatible sets of colored vertices, we have $V_1 \cap V_2 = \varnothing$, so that $V_1$ and $V_2$ form a partition of the vertex set $V = V_1 \cup V_2$. From the definition of $E$, each edge $(v, \sigma_c(v)) \in E$ is incident to one vertex in $V_1$ and one in $V_2$, which ensures that $(V, E)$ is bipartite.

A graph is edge-colored if and only if no two edges incident to the same vertex share the same color. We define the color of each edge as follows: an edge $e = (v, \sigma_c(v))$, with $c \in [d]$ and $v \in V_1|c$, is assigned color $c$. Proposition~\ref{prop:edges} allows us to write:
\bea
\bigcup_{v \in V_1} \left\{ (v, \sigma_c(v)) \mid c \in \phi_1(v) \right\} = E = \bigcup_{v \in V_2} \left\{ (\sigma_c^{-1}(v), v) \mid c \in \phi_2(v) \right\}.
\eea
The type function $\phi$ guarantees that, pairwise, the half-edges incident to the same vertex have distinct colors.
\end{proof}

We have introduced definitions that are compatible with both graph theory and tensor contractions. The next step is to embed these constructions within
a single framework.

\subsection{Tensor contractions of multiple-orders}\label{subsec:defMulOrdContrac}
This section is devoted to the formulation of tensor contractions of arbitrary order. The formalism will be developed in close connection with the constructions introduced in the previous section.

The following statement holds: 
\begin{proposition}[Equivalence classes of vertices]
\label{def:typeEquivRelation} 
Let $(V, \phi)$ be a colored set of vertices of chromatic index $d \in \mathbb{N}^*$. The relation on the vertex set $V$ defined by
\bea
\forall v, w \in V, \;\; v \, R_T \, w \;\Leftrightarrow\; \phi(v) = \phi(w)
\eea
is an equivalence relation.
\end{proposition}
The proof of the above proposition is direct.
This statement informs us that two vertices are equivalent if and only if they share the same type. The  corresponding quotient set $V/R_T$ is 
\bea V/R_T  = (\phi^{-1}(A))_{A \in P([d])}
\eea

\begin{definition}[Tensor associated with a colored vertex]
Let $(V, \phi)$ be a colored set of vertices of chromatic index $d \in \mathbb{N}^*$. Let $E$ be a complex vector space of dimension $N$. To each $v \in V$, we assign a tensor
\[
T^{(v)} : E^{\times |\phi(v)|} \longrightarrow \mathbb{C}
\]
of order $p = |\phi(v)|$. In a fixed basis of $E$, if $\phi(v) = \{c_1, \dots, c_p\}$, the components of $T^{(v)}$ are denoted by $T^{(v)}_{i_{c_1} \cdots i_{c_p}}$, where for all $k \in [p]$, we have $i_{c_k} \in [N]$, and $i_{c_k}$ is the index corresponding to the color $c_k \in \phi(v)$.

We refer to the collection $(T^{(v)})_{v \in V}$ as a colored family of tensors associated with the colored set of vertices $(V, \phi)$.
\end{definition}

 Consistently with proposition \ref{def:typeEquivRelation}, one  defines an equivalence relation in the colored family of
 tensors associated with $(V,\phi)$: 
\bea \label{eq:tensorEquivalence}
\forall v,w \in V, \quad \, T^{(v)} \, R_T \, T^{(w)} \Leftrightarrow v \, R_T \, w.
\eea
This means that two tensors are considered equivalent if their associated vertices share the same color type.

In what follows, equivalent tensors will be regarded as indistinguishable and hence identified: 
\bea \label{eq:indistinguishableTensor} 
 T^{(v)} \, R_T \, T^{(w)} \Rightarrow T^{(v)} = T^{(w)}. 
\eea
This convention is consistent with definition~\ref{def:ContrcUnitaireMonoTenseur}, where fixed-order tensor contractions were constructed using indistinguishable tensors.

We now address a notational consideration. Recall the notation for tensor components, $T^{(v)}_{i_{c_1} \cdots i_{c_p}}$. Since, in a tensor contraction, all indices are summed over, it is not necessary to keep track of each individual index $i^{(v)}_c \in [N]$ corresponding to the color $c \in [d]$ for a given tensor $T^{(v)}$. As these are dummy summation variables, only the color label $c$ will be relevant. Therefore, to simplify the notation, we adopt the following convention:
\bea 
T^{(v)}_{i_{c_1} \cdots i_{c_p}} \equiv T^{(v)}_{i, \phi(v)} 
\eea

\begin{definition}[Multiple-order tensor contraction]\label{def:CtrOrdreMult} Let $G= (\Lambda, \Gamma, \s)$ be a colored bipartite  graph of chromatic index $d \in \N^*$,  with $\Lambda = (V, \phi)$, $\Gamma = (W, \psi)$, and $\s = (\s_1, \dots, \s_d) \in S = \bigtimes_{c= 1}^d S(c)$. 
Let $(T^{(v)})_{v \in V}$ be a colored family of tensors associated with $\Lambda$, 
and $(R^{(w)})_{w \in W}$ be a colored family of tensors associated with $\Gamma$.  
A multiple-order contraction of  
$(T^{(v)})_{v \in V}$ and $(R^{(w)})_{w \in W}$
parametrized by the graph $G$ is defined by  
\bea 
I(G; (T^{(v)})_{v \in V} , (R^{(w)})_{w \in W}) = K(G) \left( \prod_{v \in V} T^{(v)}_{i, \phi(v)}\right) \left( \prod_{w \in W} \ov{R}^{(w)}_{j, \psi(w)}\right)\,,
\eea
with the contraction kernel given by:
\bea K(G) = \prod_{c = 1}^d \prod_{v \in V|c} \delta(i_c^{(v)}, j_c^{(\s_c(v))}) \,, 
\eea 
and where
    $i^{(v)}_c$ and  $j_c^{(w)}$ refer to summation indices of color $c$ of  $T^{(v)}$ and   
$R^{(w)}$, respectively.   
\end{definition}
We observe that the contraction kernel $K(G)$ is implicitly indexed by the set $E$ of edges defined in proposition~\ref{prop:linkWithCOnventionalGraph}. For each edge $(v, \sigma_c(v)) \in E$, the corresponding pair of upper indices in the pair $(i_c^{(v)}, j_c^{(\sigma_c(v))})$ forms the labeling associated with that edge.

\begin{proposition} \label{prop:particularCaseContraction} 
Let $G = (\Lambda, \Gamma, \sigma)$ be a colored bipartite graph of chromatic index $d \in \mathbb{N}^*$, where $\Lambda = (V, \phi)$, $\Gamma = (W, \psi)$, and $\sigma = (\sigma_1, \dots, \sigma_d) \in S := \bigtimes_{c=1}^d S(c)$. Let $(T^{(v)})_{v \in V}$ be a colored family of tensors associated with $\Lambda = (V, \phi)$, and  $(R^{(w)})_{w \in W}$ be a colored family of tensors associated with $\Gamma = (W, \psi)$. We assume the following conditions hold:
\bea
\label{eq:ConstraintContraction}
&&(i)  \quad \forall (v, w) \in V \times W, \quad \phi(v) = [d] = \psi(w), \crcr
&&(ii)  \quad \exists (v, w) \in V \times W \text{ such that } T^{(v)} = R^{(w)}.
\eea
Under these assumptions, both colored families of tensors $(T^{(v)})_{v \in V}$ and $(R^{(w)})_{w \in W}$ reduce to a single set of $n = |V| = |W|$ identical tensors $T$ of order $d$. 
In this setting, we have:
\bea
I\left(G;\, (T^{(v)})_{v \in V}, (R^{(w)})_{w \in W} \right) = I(\sigma; T),
\eea
where $I(\sigma; T)$ denotes the contraction defined in definition~\ref{def:ContrcUnitaireMonoTenseur}.

\end{proposition}
\begin{proof}
The constraint $(i)$ of \eqref{eq:ConstraintContraction} indicates that all vertices of $V$ ($W$, respectively) are of the same type and of valence $d$, carrying all the $d$ colors on their half-edges. 
Then, by definition: $\forall c \in [d],  V|c = V$
($W|c = W$, respectively). 

The compatibility condition between $\Lambda$ and $\Gamma$ implies that their color multiplicities are equal, and therefore that $ |V|=|W|$. 
We introduce the integer $n = |V| = |W|$ and 
a label for each vertex of $V = \{v_1, \dots , v_n\}$ and $W=\{w_1, \dots, w_n\}$. 
As a result, for all $c$ in $[d]$, $S(c)$ is the set of bijections between $V|c=V$ and $W|c=W$, two sets of cardinality $n$, giving $S(c) \cong S_n$ and $\bigtimes_{c = 1}^d S(c) \cong S_n^{\times d}$.
 Consider two vertices $v_i \in V$ and $w_j \in W$ such that $\sigma_c(v_i) = w_j$ for some $\sigma_c \in S(c)$. We may then identify a corresponding permutation $\tilde{\sigma}_c \in S_n$ satisfying $\tilde{\sigma}_c(i) = j$, \textit{via} the isomorphism established above. For simplicity, we will use the same notation and write $\sigma_c = \tilde{\sigma}_c$.

Since all tensors $T^{(v)}$, for $v \in V$, are equivalent, they are identified with a single tensor, denoted by $T$. Similarly, all tensors $R^{(w)}$, for $w \in W$, are equal to a common tensor $R$. Consequently, the contraction 
$I(G;(T^{(v)})_{v \in V} , (R^{(w)})_{w \in W})$  can now be expressed as:
\bea 
I(G; (T^{(v)})_{v \in V} , (T^{(w)})_{w \in W}) 
= K(G)\prod_{p = 1}^n T_{i, \phi(v_p)} \prod_{q=1}^n \ov{T}_{j, \psi(w_q)}, 
\eea
Inspecting the contraction kernel, we have:
\bea 
K(G) &=& \prod_{c = 1}^d \prod_{v \in V|c} \delta(i_c^{(v)}, j_c^{(\s_c(v))}) 
= \prod_{c = 1}^d \prod_{v \in V} \delta(i_c^{(v)}, j_c^{(\s_c(v))})\crcr    
&=&  \prod_{c = 1}^d \prod_{p = 1}^n \delta(i_c^{(v_p)}, j_c^{(\s_c(v_p))})
     = 
 \prod_{c = 1}^d \prod_{p = 1}^n \delta(i_c^{(p)}, j_c^{(\s_c(p))})
     \crcr
&=& K(\s; \{i_1^{(j)}, i_2^{(j)},\dots , i_k^{(j)}\}_j, \{i_1^{(k)}, i_2^{(k)},\dots , i_k^{(k)}\}_k)
\eea
where the last expression is the tensor contraction kernel of definition \ref{def:ContrcUnitaireMonoTenseur}. 
\end{proof}  

In the following, we introduce the tuple 
and use at times
\bea \label{eq:TupleContraction}
I= (V, \phi, W, \psi, \s, (T^{(v)})_{v \in V}, (R^{(w)})_{w \in W})
\eea
to designate the  multiple-order tensor contraction 
$I(G;  (T^{(v)})_{v \in V} , (R^{(w)})_{w \in W})$ of definition \ref{def:CtrOrdreMult}. 

\subsection{Unitary and group action invariances}\label{subsec:UniAndGroupAction}

This subsection aims at proving that multiple-order contractions of complex tensors, as stated in definition \ref{def:CtrOrdreMult}, have two invariances. First, they are invariant under the unitary group, this is discussed in the following section. Secondly, they are invariant under a group action, which generalizes the group action in  definition \ref{def:GroupActionMonoOrder}.

\subsubsection{Invariance under the unitary group} \label{subsubsec:UnitaryGroup}

\begin{theorem}[Unitary invariance] \label{theo:InvarianceUnderU}Let $I= (V, \phi, W, \psi, \s, (T^{(v)})_{v \in V}, (R^{(w)})_{w \in W})$ be a multiple-order contraction.  Let $N$ be the dimension of the $\C$-vector space $E$ where are defined the following families of tensors. Then, $I(G; (T^{(v)})_{v \in V} , (R^{(w)})_{w \in W})$  is invariant under the transformation of each tensor $T^{(v)}$
under the fundamental action of the group $U(N)^{\otimes |\phi(v)|}$, 
and each tensor $R^{(w)}$
under the fundamental action of the group $U(N)^{\otimes |\psi(w)|}$.
\end{theorem}
\begin{proof}
Let $U=[U_{ij}]_{1 \le i,j\le N}\in U(N)$ be a unitary matrix.  
Let $v \in V$, the type of $T^{(v)}$ is given by $\phi(v) = \left\{a_1, a_2, \dots, a_k \right\}$ where $a_i \in [d]$. The coefficients of $T^{(v)}$
in a given base 
are denoted by $T^{(v)}_{i_{a_1}, \cdots, i_{a_k}}$. 
Let 
$T'^{(v)}_{i_{a_1}, \cdots, i_{a_k}}$ denote
the coefficients after unitary transformation. The link between both is given by definition \ref{def:tenseur} (see in particular equation \eqref{eq:chtTenseurBase}): 
\bea 
T'^{(v)}_{i'_{a_1}, \cdots, i'_{a_k} } = U_{i'_{a_1}i_{a_1}}U_{i'_{a_2}i_{a_2}}\cdots \, U_{i'_{a_k}i_{a_k}} T^{(v)}_{i_{a_1}, \cdots, i_{a_k}}
\eea
Under the same convention, for every $w$ in $W$, we have the unitary transformation rule for the coefficients of $R^{(w)}$ given by: 
\bea 
R'^{(w)}_{j'_{b_1}, \cdots, j'_{b_m} } = U_{j'_{b_1}j_{b_1}}U_{j'_{b_2}j_{b_2}}\cdots \, U_{j'_{b_m}j_{b_k}} R^{(w)}_{j_{b_1}, \cdots, j_{b_m}},
\eea
where $\psi(w) = \{b_1,b_2, \dots, b_m\}$.
We denote by $I'\equiv I'(G; (T^{(v)})_{v \in V} , (R^{(w)})_{w \in W})$ the contraction formed from the base $\mathcal{B'}$ and by $I \equiv I(G; (T^{(v)})_{v \in V} , (R^{(w)})_{w \in W})$ the contraction formed from the base $\mathcal{B}$. There contraction kernels are respectively $K'(G)$ and $K(G)$. Then,  
\bea
\begin{aligned}
    I' &= K'(G) \left( \prod_{v \in V} T'^{(v)}_{i', \phi(v)}\right) \left( \prod_{w \in W} \ov{R'}^{(w)}_{j', \psi(w)}\right)\\
    &= K'(G) \left(\prod_{v \in V} \prod_{c \in \phi(v)} U_{i'^{(v)}_ci_c^{(v)}}\right)\left( \prod_{v \in V} T^{(v)}_{i, \phi(v)}\right) \left(\prod_{w \in W} \prod_{c \in \psi(w)} \ov{U}_{j'^{(w)}_c j_c^{(w)}}\right)\left( \prod_{w \in W} \ov{R}^{(w)}_{j, \psi(w)}\right)
\end{aligned}
\eea
By virtue of proposition \ref{prop:edges},  the product over the matrix coefficients can be written as:
\bea \label{eq:proofUnitarity1}\prod_{v \in V} \prod_{c \in \phi(v)} U_{i'^{(v)}_ci_c^{(v)}} = \prod_{c = 1}^d \prod_{v \in V|c} U_{i'^{(v)}_ci_c^{(v)}}\eea
\bea \label{eq:proofUnitarity2}\prod_{w \in W} \prod_{c \in \psi(w)} \ov{U}_{j'^{(v)}_c j_c^{(v)}}  = \prod_{c = 1}^d \prod_{w \in W|c} \ov{U}_{j'^{(w)}_c j_c^{(w)}} = \prod_{c = 1}^d \prod_{v \in  V|c} \ov{U}_{j'^{(\s_c(v))}_c j_c^{(\s_c(v))}}.\eea

Indeed, in equation \eqref{eq:proofUnitarity1}, we used the fact that each term of the product was in bijection with an element of the edge set \eqref{eq:edgesSet} to re-order the product. The same reasoning holds for equation \eqref{eq:proofUnitarity2} with at the end the change of variable $v = \s_c^{-1}(w)$ since $V|c$ is in bijection with $W|c$ through $\s_c$. As a consequence, equations \eqref{eq:proofUnitarity1} and \eqref{eq:proofUnitarity2} can be merged with the contraction kernel $K'(G)$ so as to have:
\bea I'(G) = L(G)\left( \prod_{v \in V} T^{(v)}_{i, \phi(v)}\right)\left( \prod_{w \in W} \ov{R}^{(w)}_{j, \psi(w)}\right)\eea
with 
\bea  L(G) = \prod_{c = 1}^d \prod_{v \in V|c}
 U_{i'^{(v)}_ci_c^{(v)}} \delta(i_c^{'(v)}, j_c^{'(\s_c(v))})\ov{U}_{j'^{(\s_c(v))}_c j_c^{(\s_c(v))}}\eea
 Noticing that $i^{'(p)}_c$ and $j^{'(\s_c(p))}_c$ are free indices of summation, one can integrate the Kroneckers:
 \bea
L(G) = \prod_{c = 1}^d \prod_{v \in V|c}
 U_{i'^{(v)}_ci_c^{(v)}}\ov{U}_{i'^{(v)}_c j_c^{(\s_c(v))}}
 = \prod_{c = 1}^d \prod_{v \in V} \delta(i_c^{(v)}, j_c^{(\s_c(v))}).
 \eea
 The last equality uses the unitarity condition on the matrices: $\ov{U}_{ba}U_{bc} = \delta(a, c)$. Ultimately, $L(G) = K(G)$ and $I' = I$. 
\end{proof}

\subsubsection{Invariance under group action}
\label{subsubsec:InvarianceUnderPermutation}

We reveal a second invariance of multiple-order tensor contractions determined by the indistinguishability of the tensors at fixed order and given color type. 

Let $V$ be a finite set 
and $U \subseteq V$ be a subset of $V$. 
Let $\s\in S(V)$,  where 
$S(V)$ is the set of bijections over $V$. 
$S(V)$ is isomorphic to the  
group $S_{|V|}$ of permutations of 
$|V|$ objects. 
If $\s (U) \subset U$, then the restriction 
of $\s$ to $U$, that we will denote by 
$\s[U]$, is a permutation of $U$. 

\begin{proposition} \label{prop:H}
Let $\Lambda = (V, \phi)$ be a colored set of vertices. Then,
\bea H(\Lambda) = \left\{ \s \in S(V) \, |\, \forall A \in P([d]), \, \s(\phi^{-1}(A)) \subseteq \phi^{-1}(A)  \right\}\eea
is a subgroup of $S(V)$. Furthermore, 
\bea H(\Lambda) = \{\s \in S(V)\, |\, \forall v \in V, \, v \, R_T \, \s(v)\}.\eea
\end{proposition}
\begin{proof}
First and foremost, $H(\Lambda) \subseteq S(V)$ and contains the identity element. Then $H(\Lambda)$ is a subgroup of $S(V)$ if and only if
for two $\s, \r \in H(\Lambda)$, $\s\r^{-1} \in H(\Lambda)$. 
One should notice that:
\bea \label{eq:proofSubgroup}\forall \s \in H(\Lambda), \forall A \in P([d]), \, \s(\phi^{-1}(A)) = \phi^{-1}(A) \Leftrightarrow \s^{-1}(\phi^{-1}(A)) = \phi^{-1}(A)
\eea
Let $s, \r$ be two elements of $H(\Lambda)$. It is clear with equation \eqref{eq:proofSubgroup} that $\r^{-1} \in H(\Lambda)$ and
\bea 
\forall A \in P([d]),\,  \r^{-1}(\phi^{-1}(A)) \subseteq \phi^{-1}(A)\Rightarrow \s\r^{-1}(\phi^{-1}(A)) \subseteq \phi^{-1}(A). 
\eea
As a result, $H(\Lambda)$ is a subgroup of $S(V)$.

Now, we define $X = \{\s \in S(V)\, | \, v \, R_T \, \s(v)\}$. Let $\s \in X$, $A \in P([d])$ and $v \in V$ such that $\phi(v) = A$. Since $v \, R_T \, \s(v)$, we have $\phi(\s(v)) = \phi(v) = A$ and $\s(v) \in \phi^{-1}(A)$. As a result $X \subseteq H(\Lambda)$.  Conversely, let $\s \in H(\Lambda)$ and $v \in V$. There exists $A \in P([d])$ such that $\phi(v) = A$ and, by definition of $H(\Lambda)$, $\s(v) \in \phi^{-1}(A)$ \textit{i.e.} $\phi(\s(v)) = A = \phi(v)$. Therefore, $v \, R_T\, \s(v)$ and $H(\Lambda) \subseteq X$. Therefore, $H(\Lambda) = X$. 
\end{proof}

\begin{proposition}[Decomposition of the sugroup action]\label{prop:decompositionPermutation}
Let $\Lambda = (V, \phi)$ be a colored set of vertices. Then, 
 $H(\Lambda)$ is isomorphic to $\bigtimes_{A \in P([d])} S(\phi^{-1}(A))$. Furthermore, every permutation $\pi$ in $H(\Lambda)$ admits a decomposition into a product of disjoint permutations:
\bea
\label{eq:piProdpi}
\pi = \prod_{A \in P([d])} \tilde{\pi}_A,
\eea 
where for all $A \in P([d])$, $\pi_A = \pi[\phi^{-1}(A)] \in S(\phi^{-1}(A))$ and $\tilde{\pi}_A$ is an extension of $\pi_A$ over $V$ such that $\tilde{\pi}_A[V \backslash \phi^{-1}(A)]$ is equal to the identity.
\end{proposition}

\begin{proof} Let $A \in P([d)])$ and $\pi \in H(\Lambda)$. Since $\pi(\phi^{-1}(A)) \subseteq \phi^{-1}(A)$, we define the restriction $\pi_A \equiv \pi[\phi^{-1}(A)]: \phi^{-1}(A) \longrightarrow \phi^{-1}(A)$
and its extension  $\pi_A$ over $V$ by demanding that $\tilde{\pi}_A = id$ on $V \backslash \phi^{-1}(A)$. Since $(\phi^{-1}(A))_{A \in P([d])}$ is a partition of $V$, and $\forall A\neq B \in P([d]),\, \tilde{\pi}_A\tilde{\pi}_B = \tilde{\pi}_B\tilde{\pi}_A$. The decomposition into disjoint permutations
\eqref{eq:piProdpi} follows from the fact
that the union over $A\in P([d])$ of $\phi^{-1}(A)$ covers $V$. 

One defines the mapping $\r$ by :   
\bea \r : \left\{ \begin{array}{l l}
 H(\Lambda) \rightarrow \bigtimes_{A \in P([d])} S(\phi^{-1}(A)) \\
\pi \mapsto (\pi_A)_{A \in P([d])}
\end{array}\right. \eea
It is then direct to show that $\r$ is an isomorphism. 

\end{proof}

Let $\Lambda = (V, \phi)$ be a colored set of vertices of chromatic index $d \in \N^*$. Then, 
\bea 
\forall c \in \inter{1}{d}, \, \forall \s \in  H(\Lambda),\, \s\left(V|c\right) \subseteq V|c \; . 
\eea
The proof of this statement is simple. Since 
\bea 
\forall c \in P([d]), \, V|c = \{ v \in V\, | \, c \in \phi(v) \} = \bigcup_{A \in P([d])\, | \, c \in A}\phi^{-1}(A)\,, 
\eea
the property of elements of $H(\Lambda)$ allows us to write 
\bea 
\s\left( V|c \right)  = \bigcup_{A \in P([d])\, | \, c \in A} \s \left( \phi^{-1}(A) \right) \subseteq \bigcup_{A \in P([d])\, | \, c \in A} \phi^{-1}(A) = V|c \, . 
\eea

\begin{definition}[Group action]\label{def:groupActionGeneralized} Let $\Lambda= (V, \phi)$, $\Gamma=(W, \psi)$ be two compatible colored sets of vertices of chromatic index $d \in \N^*$. We define a group action of $H(\Lambda)\times H(\Gamma)$ on the set $S=\bigtimes_{c = 1}^d S(c)$ as: 
\bea 
\alpha(\Lambda, \Gamma): \left\{ \begin{array}{l l}
H(\Lambda)\times H(\Gamma) \times S \rightarrow S\\
(\pi, \eta, \s) \mapsto \eta.\s.\pi \equiv \left( \eta_1 \, \s_1 \, \pi_1,\eta_2 \, \s_2 \, \pi_2, \dots,   \eta_d \, \s_d \, \pi_d\right)
\end{array}\right.
\eea
with the notation $\pi_c \equiv \pi[V|c]$ and  $\eta_c \equiv \eta[W|c]$, for any color $c \in \inter{1}{d}$.

\end{definition}

\begin{proposition} \label{prop:paritcularGroupAction} Let $G= (\Lambda, \Gamma, \s)$ be a colored bipartite graph of chromatic index $d \in \N^*$ with $\Lambda = (V, \phi)$, $\Gamma = (W, \psi)$ and $\s = (\s_1, \dots, \s_d) \in S = \bigtimes_{c= 1}^d S(c)$. We assume the condition $(i)$ of proposition \ref{prop:particularCaseContraction} holds. 
Then, 
\bea \alpha(\Lambda, \Gamma) = \alpha^d_n  \,, 
\eea
where $n \equiv |V| = |W|$, and $\alpha^d_n$ is the group action of definition \ref{def:GroupActionMonoOrder}.
\end{proposition}
\begin{proof}
With the condition $(i)$ of proposition \ref{prop:particularCaseContraction}, we have for all $c$ in $[d]$: $V|c = V$ and similarly, $W|c = W$. Thus, using the the compatibility between $\Lambda$ and $\Gamma$ and the equality of color multiplicity, we have $|W| = |V|$. The cardinalities $|V|=|W|$ are noted down $n$. On the other hand, since for all $c \in [d]$, $S(c)$ is the set of bijection between $W|c = W$ and $V|c = V$, it is clear that $S(c)$ is isomorphic to $S_n$ the symmetric group and $S = \bigtimes_{c = 1} S(c)\cong S_n^{\times d}$. Similarly, $S(V)$, the set of bijections between vertices of $V$ is also isomorphic to $S_n$, and likewise $S(W) \cong S_n$. Then, $H(\Lambda)$ becomes 
\bea
H(\Lambda)
= \left\{ \s \in S(V) \, |\, \s(\phi^{-1}([d])) \subseteq \phi^{-1}([d])  \right\} 
= S(V) \cong S_n \, . 
\eea
We find the same results concerning $\Gamma$, $H(\Gamma) \cong S_n$. Using definition \ref{def:groupActionGeneralized}, for all $c$ in $[d]$, for all $\pi$ in $S(V)$ and $\eta$ in $S(W)$:
\bea 
\pi_c \equiv \pi[V|c] = \pi \quad \text{ and } \quad   \eta_c \equiv \eta[W|c] = \eta \,. 
\eea
Finally, $\alpha(\Lambda, \Gamma)$ meets the definition \ref{def:GroupActionMonoOrder} of $\alpha^d_n$, up to group isomorphisms 
$H(\Lambda) \cong S_n \cong H(\Gamma)$, 
and $S \cong S_n^{\times d}$. 
\end{proof}

\begin{theorem}\label{theo:InvGroupCtrcMultiple}
Let $ I(\Lambda,\Gamma, \s; (T^{(v)})_{v \in V} , (R^{(w)})_{w \in W})$ be a multiple-order contraction. Then, $\forall \pi \in H(\Lambda), \forall \eta \in H(\Gamma)$, 
\bea 
 I(\Lambda,\Gamma, \eta.\s.\pi; (T^{(v)})_{v \in V} , (R^{(w)})_{w \in W}) =   I(\Lambda,\Gamma, \s; (T^{(v)})_{v \in V} , (R^{(w)})_{w \in W}) \,. 
\eea

\end{theorem}

\begin{proof}
Consider  
$I= (\Lambda=(V, \phi), \Gamma=( W, \psi), \s, (T^{(v)})_{v \in V}, (R^{(w)})_{w \in W})$ a multiple-order contraction. 
A key argument in the proof is given in  equations \eqref{eq:tensorEquivalence} and \eqref{eq:indistinguishableTensor}. For $v_1$ and $v_2$ in $V$, $T^{(v_1)}$ and $T^{(v_2)}$, two tensors  of order $n = \phi(v_1) = \phi(v_2)$, are stated to be equal if and only if $v_1$ and $v_2$ are equivalent. Under this condition, $T^{(v_1)}$ and $T^{(v_2)}$ can swap their labels and the set of colored indices $\{i^{(v_1)}_c\}_{c \in \phi(v_1)}$ of $T^{(v_1)}$ is swapped with $\{i^{(v_2)}_c\}_{c \in \phi(v_2)}$, the set of colored indices of $T^{(v_2)}$:
\bea
\{i^{(v_1)}_c\}_{c \in \phi(v_1)} \leftrightarrow \{i^{(v_2)}_c\}_{c \in \phi(v_2)} \Leftrightarrow v_1 \, R_T \, v_2.\eea
Then, it is possible to permute labels $(v)$ of the indices $i^{(v)}_c$ in the contraction kernel $K(G)$ with a permutation $\pi \in S(V)$ if and only if 
\bea 
\forall v \in V, \, 
\pi(v) \, R_T \, v \Leftrightarrow \pi \in H(\Lambda),
\eea
where proposition \ref{prop:H} yields the equivalence. Of course, this discussion is also true for $W$ and the indices $j^{(w)}_c$ of $(R^{(w)})_{w \in W}$, \textit{i.e.} labels $(w)$ of the indices $j^{(w)}_c$ are permutable by $\eta \in S(W)$ if and only if $\eta \in H(\Gamma)$. 

Focusing only on the contraction kernel $K(G) = K(\Lambda, \Gamma, \s)$, we consecutively apply  the permutation $\pi$ and then $\eta$ on the labels $i^{(v)}_c$ and $j^{(w)}_c$, respectively, for $v \in V$ and $w \in W$:
\bea
K(\Lambda, \Gamma, \s) = \prod_{c = 1}^d \prod_{v \in V|c}  \delta(i_c^{(\pi_c(v))}, j_c^{(\s_c(v))})
= \prod_{c = 1}^d \prod_{v \in V|c}  \delta(i_c^{(\pi_c(v))}, j_c^{(\eta_c\s_c(v))}) \,,
\eea
where, for every $c \in [d]$,  $\pi$ and $\eta$ have been restricted to $\pi_c \equiv \pi[V|c]$ and $\eta_c \equiv \eta[W|c]$, 
respectively. We can justify this by noting that $v \in V|c,\, \s_c(v) \in W|c$ and $\pi(V|c) \subseteq V|c, \, \eta(W|c) \subseteq W|c$.

A final step consists of the following change of variable: $v' = \pi(v) \Leftrightarrow v = \pi^{-1}(v')$, giving 
\bea 
K(\Lambda, \Gamma, \s) = \prod_{c = 1}^d \prod_{v' \in V|c}  \delta(i_c^{(v')}, j_c^{(\eta_c\s_c\pi_c^{-1}(v'))}) = K(\Lambda, \Gamma, \eta.\s.\pi^{-1}) 
\eea
which ends the proof. 

\end{proof}

Theorem \ref{theo:InvGroupCtrcMultiple} 
teaches us that $H(\Lambda)$ is the symmetry of group of the colored set of vertices $\Lambda$ and $H(\Gamma)$ the symmetry group of $\Gamma$.

\subsection{Counting unitary invariants}\label{subsec:ComptageInvMultiOrder}

We are in position to 
enumerate  multiple-order tensor contractions. We must take into account the orbits of the group action  previously defined. 

\begin{definition}[Cardinality function] Let $\Lambda = (V, \phi)$ be a colored set of vertices of chromatic index $d \in \N^*$. The cardinalilty function $n$ associated with $\Lambda$ is an application
$n: P([d]) \rightarrow  \N$
that maps a color type $A \in P([d])$ to the cardinality  $|\phi^{-1}(A)|$.
\end{definition}

\begin{theorem}\label{theo:CountMultiOrder}
Let $G= (\Lambda, \Gamma, \s)$ be a colored bipartite graph of chromatic index $d \in \N^*$. Let $n$ and $m$ be  the cardinality functions of $\Lambda$ and $\Gamma$, respectively. We denote by $Z(\Lambda, \Gamma)$ the number of orbits of the group action $\alpha(\Lambda, \Gamma)$. Then,
\bea \label{eq:orbitalCounting}
Z(\Lambda, \Gamma) = \sum_{\substack{(\mu_A)_{A\in P([d])} \, |\,  \mu_A \vdash n(A)\\ (\nu_A)_{A \in P([d])}\, | \, \nu_A \vdash m(A)}} \lrgsize{\delta}((\mu_A)_{A \in P([d])}, (\nu_A)_{A \in P([d])}) \frac{\prod_{c = 1}^d \Sym(\sum_{A \in F(c)} \mu_A)}{\prod_{A \in P([d])} \Sym(\mu_A)\Sym(\nu_A)},\eea
where: 
\begin{enumerate}
    \item[-] the sum over $\mu_A$ ($\nu_A$, resp.) denotes the sum over partitions  of $n(A)$ ($m(A)$, resp.) for a given $A \in P([d])$.
    \item[-]$\forall c \in [d], ~ F(c)= \{ A \in P([d]) ~|~ c \in A \}$.
    \item[-] $\Sym(\mu)$ is the symmetry factor associated to the integer partition $\mu$ (the cardinality of the stabilizer associated with the conjugacy class labeled by $\mu$). 
    \item[-] $\lrgsize{\delta}((\mu_A)_{A \in P([d])}, (\nu_A)_{A \in P([d])}) \equiv \prod_{c = 1}^d \delta(\sum_{A \in F(c)} \mu_A, \sum_{A \in F(c)} \nu_A)$, each $\delta$ being 1 when the integer partitions (or sum of partitions) are equal and 0 otherwise.
\end{enumerate}
\end{theorem}
\begin{proof}
Let $\Lambda = (V, \phi)$ and $\Gamma = (W, \psi)$ two colored sets of vertices, and denote by $S=\bigtimes_{c = 1}^d S(c)$, the finite set on which the group $H(\Lambda) \times H(\Gamma)$ acts on. 
The Burnside lemma enables us to rewrite
the number of orbits of this action as: 
\bea 
\begin{aligned}
Z(\Lambda, \Gamma) &= \frac{1}{|H(\Lambda)\times H(\Gamma)|} \sum_{(\pi, \eta) \in H(\Lambda) \times H(\Gamma)} \left| S^{(\pi, \eta)} \right| \\
&= \frac{1}{|H(\Lambda)||H(\Gamma)|} \sum_{(\pi, \eta) \in H(\Lambda) \times H(\Gamma)} \sum_{\s \in S} \Delta(\eta.\s.\pi, \s) 
\end{aligned}
\eea 
where $S^{(\pi, \eta)}= \left\{ \s \in S \, | \, \eta.\s.\pi  = \s \right\}$ and $\Delta(\eta.\s.\pi, \s) = 1$ if $\eta.\s.\pi = \s$, 0 if not. We note that $\Delta(\eta.\s.\pi, \s)$ factorizes into a product of $\Delta$ on each factor subgroups as 
\bea 
\Delta(\eta.\s.\pi, \s) = \prod_{c = 1}^d \Delta(\eta_c \s_c \pi_c, \s_c) 
\eea
and infer 
\bea 
\label{eq:ZBurnLemm}
\begin{aligned}
Z(\Lambda, \Gamma) &= \frac{1}{|H(\Lambda)\times H(\Gamma)|} \sum_{(\pi, \eta) \in H(\Lambda) \times H(\Gamma)} \prod_{c = 1}^d \left( \sum_{\s_c \in S(c)} \Delta(\eta_c \s_c \pi_c, \s_c)\right).
\end{aligned}
\eea
Bearing in mind proposition \ref{prop:decompositionPermutation}, $H(\Lambda)$ (resp. $H(\Gamma)$) is isomorphic to $\bigtimes_{A \in P([d])} S(\phi^{-1}(A))$ (resp. $\bigtimes_{A \in P([d])} S(\psi^{-1}(A))$), in such a way that $\pi$ and $\eta$ admit a decomposition into disjoint permutations as: 
\bea 
\label{theo:proofCounting1}
\pi = \prod_{A \in P([d])} \tilde{\pi}_A \quad \text{ and }\quad  \eta = \prod_{A \in P([d])}
\tilde{\eta}_A  \,.  
\eea
For all $c$ in $[d]$, $\pi$ (resp. $\eta$) restricts to $V|c$ (resp. $W|c$) 
in the following way: 
\bea \label{theo:proofCounting2}
\pi_c = \prod_{A \in F(c)} \tilde{\pi}_A[V|c]
\quad (\text{ resp. } \eta_c = \prod_{A \in F(c)}
\tilde{\eta}_A [W|c]\; ) \;.
\eea
It is therefore possible to reformulate the sum over $H(\Lambda) \times H(\Gamma)$ as
\bea 
\begin{aligned}
Z(\Lambda, \Gamma) &= \frac{1}{|H(\Lambda)||H(\Gamma)|} \sum_{\pi_A \in S(\phi^{-1}(A))} \;\;\sum_{\eta_A \in S(\psi^{-1}(A))} \times \\
&\times \prod_{c = 1}^d \sum_{\s_c \in S(c)} \Delta(\prod_{A \in F(c)} \tilde{\eta}_A[W|c]\, \s_c\, \prod_{A \in F(c)} \tilde{\pi}_A[V|c], \s_c) 
\end{aligned}
\eea
We denote, for all $c$ in $[d]$,  by $\mu_c$ (resp. $\nu_c$) the integer partition associated to the conjugation class of $\pi_c$ (resp. $\eta_c$). Looking at the decomposition in disjoint permutations provided by equation \eqref{theo:proofCounting2}, we have 
\bea \label{theo:proofCounting3}
\mu_c = \sum_{A \in F(c)} \mu_A \quad  \text{ and } \quad \nu_c = \sum_{A \in F(c)} \nu_A,\eea
where for all $A$ in $P([d])$, $\mu_A \vdash n(A)$ (resp. $\nu_A \vdash m(A)$) is the integer partition associated to the conjugation class of $\pi_A$ (resp. $\eta_A$). Therefore, one can notice that,
\bea
\forall c \in [d], \, \Delta(\eta_c \s_c \pi_c, \s_c) = \Delta(\eta_c \s_c \pi_c, \s_c) \delta\left( \mu_c, \nu_c\right).
\eea
Indeed for $c$ in $[d]$, 
\bea \Delta(\eta_c \s_c \pi_c, \s_c) = 1  \Leftrightarrow \eta_c = \s_c \pi^{-1}_c \s^{-1}_c 
\eea
which implies that $\mu_c = \nu_c$,
because $\pi_c^{-1}$ is in the same conjugation class that $\pi_c$. Therefore, for every $A \in P([d])$, we can reformulate the sum over  $S(\phi^{-1}(A))$ (resp. $S(\psi^{-1}(A))$) by a sum over conjugacy classes $(C_A(\mu_A))_{\mu_A \vdash n(A)}$ (resp. $(C_A(\nu_A))_{\nu_A \vdash m(A)}$) that partition it. Then,
 \bea 
 \begin{aligned}
 Z(\Lambda, \Gamma) &= \frac{1}{|H(\Lambda)||H(\Gamma)|} \sum_{\substack{(\mu_A)_{A\in P([d])} \, |\,  \mu_A \vdash n(A)\\ (\nu_A)_{A \in P([d])}\, | \, \nu_A \vdash m(A)}} \;\;\sum_{\substack{(\pi_A)_{A \in P([d])} \, | \, \pi_A \in C_A(\mu_A) \\ (\eta_A)_{A \in P([d])} \, | \, \eta_A \in C_A(\nu_A) }} \times  \\
 &\times \prod_{c = 1}^d \sum_{\s_c \in S(c)}  \Delta(\eta_c \, \s_c\, \pi_c, \s_c)  \delta\left( \mu_c, \nu_c\right)\\
 &= \frac{1}{|H(\Lambda)||H(\Gamma)|} \sum_{\substack{(\mu_A)_{A\in P([d])} \, |\,  \mu_A \vdash n(A)\\ (\nu_A)_{A \in P([d])}\, | \, \nu_A \vdash m(A)}} \;\;\sum_{\substack{(\pi_A)_{A \in P([d])} \, | \, \pi_A \in C_A(\mu_A) \\ (\eta_A)_{A \in P([d])} \, | \, \eta_A \in C_A(\nu_A) }} \times  \\
 &\times \prod_{c = 1}^d \delta\left( \mu_c, \nu_c\right) \sum_{\s_c \in S(c)}  \Delta(\eta_c \, \s_c\, \pi_c, \s_c).
 \end{aligned}
 \eea
where $\pi_c, \eta_c, \mu_c$ and $\nu_c$ are given by equations \eqref{theo:proofCounting2} and \eqref{theo:proofCounting3}. Since $\delta\left( \mu_c, \nu_c\right)$ is in prefactor of the sum over $S(c)$, we can without loss of generality consider that $\pi_c$ and $\eta_c$ are such that $\mu_c = \nu_c$. Then $\pi_c$, $\pi_c^{-1}$ and $\eta_c$ have the same cycle structure and there exists $\r: V|c \longrightarrow W|c$ such that: $\eta_c = \r \pi_c^{-1} \r^{-1}$. The last term is rewritable as 
\bea
\begin{aligned}
\sum_{\s_c \in S(c)}  \Delta( \eta_c \s_c \pi_c, \s_c) &= \sum_{\s_c \in S(c)}  \Delta( \r \pi_c^{-1} \r^{-1}\s_c \pi_c, \s_c)
= \sum_{\s_c \in S(c)}  \Delta(\pi_c^{-1} \r^{-1}\s_c \pi_c, \r^{-1}\s_c)\\
&= \sum_{\s'_c \in S(V|c)}  \Delta(\pi_c^{-1} \s'_c \pi_c, \s'_c)
= \left| \left\{ \s'_c \in S(V|c) \, | \, \s'_c \pi_c \s_c^{'-1} = \pi_c \right\}\right|\\
&= \left|\Stab(\pi_c)\right|,
\end{aligned}
\eea
where we used the fact that the map 
$S(c) \rightarrow  S(V|c)$ sending 
$\s_c \mapsto \s'_c = \r^{-1}\s_c$
is a bijection between $S(c)$ and $S(V|c)$. 
This  allows a summation over $S(V|c)$ instead of $S(c)$. $\Stab(\pi_c)$ is the stabilizer of $\pi_c$ for the inner automorphism. With the orbit-stabilizer theorem (see proposition \ref{prop:orbstab}), we have 
\bea 
|S(V|c)| = |\Stab(\pi_c)||C(\pi_c)| \Leftrightarrow |\Stab(\pi_c)| = \frac{|S(V|c)|}{|C(\pi_c)|}, 
\eea
$|\Stab(\pi_c)|$ only depends of $\mu_c$, the integer partition associated to $C(\pi_c) = C(\mu_c)$, the conjugacy class of $\pi_c$. We introduce the symmetry factor $\Sym(\mu_c)$: $\Sym(\mu_c) = |\Stab(\mu_c)|$. Using equation \eqref{theo:proofCounting3}, $Z(\Lambda, \Gamma)$ takes the form:  
\bea 
\begin{aligned}
Z(\Lambda, \Gamma) &= \frac{1}{|H(\Lambda)||H(\Gamma)|} \sum_{\substack{(\mu_A)_{A\in P([d])} \, |\,  \mu_A \vdash n(A)\\ (\nu_A)_{A \in P([d])}\, | \, \nu_A \vdash m(A)}}
\;\;\;\sum_{\substack{(\pi_A)_{A \in P([d])} \, | \, \pi_A \in C_A(\mu_A) \\ (\eta_A)_{A \in P([d])} \, | \, \eta_A \in C_A(\nu_A) }} \times \\
&\times  \Sym \left( \sum_{A \in F(c)} \mu_A \right)\prod_{c = 1}^d \delta\left( \sum_{A \in F(c)} \mu_A, \sum_{A \in F(c)} \nu_A\right)\\
&= \frac{1}{|H(\Lambda)||H(\Gamma)|} \sum_{\substack{(\mu_A)_{A\in P([d])} \, |\,  \mu_A \vdash n(A)\\ (\nu_A)_{A \in P([d])}\, | \, \nu_A \vdash m(A)}}  \lrgsize{\delta}((\mu_A)_{A \in P([d])}, (\nu_A)_{A \in P([d])})  \times \\
&\times \left(\prod_{A \in P([d])}|C_A(\mu_A)||C_A(\nu_A)|\right)\prod_{c = 1}^d \Sym \left( \sum_{A \in F(c)} \mu_A \right)\\
&= \frac{\prod_{A \in P([d])} |S(\phi^{-1}(A))||S(\psi^{-1}(A))|}{|H(\Lambda)||H(\Gamma)|} \times \\
&\times \sum_{\substack{(\mu_A)_{A\in P([d])} \, |\,  \mu_A \vdash n(A)\\ (\nu_A)_{A \in P([d])}\, | \, \nu_A \vdash m(A)}}  \lrgsize{\delta}((\mu_A)_{A \in P([d])}, (\nu_A)_{A \in P([d])}) \frac{\prod_{c = 1}^d \Sym \left( \sum_{A \in F(c)} \mu_A \right)}{\prod_{A \in P([d])} \Sym(\mu_A)\Sym(\nu_A)}.\\
\end{aligned}
\eea
The prefactor fraction equals 1 by propostion \ref{prop:decompositionPermutation}, 
thus giving  the expected formula.

\end{proof}

\begin{corollary}
Let $G= (\Lambda, \Gamma, \s)$ be a colored bipartite graph of chromatic index $d \in \N^*$ with $\Lambda = (V, \phi)$ and $\Gamma = (W, \psi)$. Let $n$ and $m$ be  the cardinality functions of $\Lambda$ and $\Gamma$, respectively. We assume condition $(i)$ of proposition \ref{prop:particularCaseContraction} holds.
Then, $|V| = |W|$, and
\bea Z(\Lambda, \Gamma) = Z^d_{|V|}, \eea
where $Z^d_{|V|}$ is given by theorem \ref{theo:JBandSajaye}. 
\end{corollary}
\begin{proof}
Since condition $(i)$ is verified, one could directly apply proposition \ref{prop:paritcularGroupAction}: $\alpha(\Lambda, \Gamma) = \alpha^d_n$. The two group actions are equal and necessarily admit the same number of orbits. 
Theorem  \ref{theo:CountMultiOrder}
tells us that this number of orbits is $Z(\Lambda, \Gamma) = Z^d_{|V|}$.

\end{proof}

It could be instructive to show how equation \eqref{eq:orbitalCounting} reduces to the counting at fixed order $d$. In other words, we can show that
\bea 
Z(\Lambda, \Gamma) = Z^d_{|V|} = \sum_{\mu ~ \vdash ~ |V|} \Sym(\mu)^{d-2} \,. 
\eea
First, observe that with condition $(i)$
\bea 
\forall A \in P([d])\backslash[d],\, 
n(A) = |\phi^{-1}(A)| = 0 = |\psi^{-1}(A)| = m(A)\,.
\eea
Proposition \ref{prop:paritcularGroupAction} entails $n([d]) = |V| = |W| = m([d])$. Equation \eqref{eq:orbitalCounting} simplifies because the partition $\mu$ of $0$ is the empty sum and $\Sym(\mu) = 1$ by convention, leading to:
\bea
\begin{aligned}
Z(\Lambda, \Gamma) 
&= \sum_{\substack{\mu ~ \vdash ~ |V| \\ \nu ~ \vdash ~ |W|}} \delta(\mu, \nu) \frac{\prod_{c = 1}^d \Sym(\mu)}{\Sym(\mu)\Sym(\nu)} 
= \sum_{\mu ~ \vdash ~ |V|} \frac{ \Sym(\mu)^d}{\Sym(\mu)\Sym(\mu)} \\
&= \sum_{\mu ~ \vdash ~ |V|} \Sym(\mu)^{d-2} \,. 
\end{aligned}
\eea
which is what we expect. 
\section{An application}
\label{sec:application}

In this section, we apply theorem \ref{theo:CountMultiOrder} to enumerate observables constructed from third-order tensors coupled with a specified number of matrices and vectors. The Python code for generating the integer sequences related to this counting is provided in appendix \ref{app:Code}. To validate the results, we conduct a combinatorial analysis on a simple yet nontrivial case.

\subsection{Some multiple-order contractions}

We define  $G= (\Lambda, \Gamma, \s)$ a colored bipartite graph of chromatic index $d = 3$ with:

\begin{enumerate}
    \item[(1)] $\Lambda = (V, \phi)$ of cardinality function $n$. $V=\{1,2,3\}$ and $\phi(1) = \phi(2) = [3]$, $\phi(3) = \{1,3\}$.
    \item[(2)] $\Gamma = (W, \psi)$ of cardinality function $m$. $W=\{1,2,3,4\}$ and $\psi(1) = \psi(2) = [3]$, $\psi(3) = \{1\}$, $\psi(4)=\{3\}$.
    \item[(3)] $\s = (\s_1, \s_2, \s_3)$ an element of $S= S(1) \times  S(2) \times  S(3)$.
\end{enumerate}

Let $(T^{(v)})_{v \in V}$ be a colored family of tensors associated with $\Lambda$, 
and $(R^{(w)})_{w \in W}$ be a colored family of tensors associated with $\Gamma$. An illustration of $G$ is given in figure \ref{fig:IllusPregraph} where vertices and their colored half-edges are about to be connected \textit{via} the colored set of bijections. Colors $1, 2$ and $3$ are arbitrarily taken to be, respectively, red, green and blue. The unitary multiple-order contraction corresponding to definition \ref{def:CtrOrdreMult}, $I(G; (T^{(v)})_{v \in V}, (R^{(w)})_{w \in W})$, is
\bea
I(G; (T^{(v)})_{v \in V}, (R^{(w)})_{w \in W}) = K(G) \,\,  T^{(1)}_{i_1i_2i_3}T^{(2)}_{i_1i_2i_3}T^{(3)}_{i_1i_3} \, \,\ov{R}^{(1)}_{j_1j_2j_3}\ov{R}^{(2)}_{j_1j_2j_3}\ov{R}^{(3)}_{j_1}\ov{R}^{(4)}_{j_3},\eea
where $T^{(1)} = T^{(2)}$, $R^{(1)} = R^{(2)}$ and the kernel is given by 
\bea 
\begin{aligned}
K(G) =&\delta(i_1^{(1)}, j_1^{(\s_1(1))}) \delta(i_1^{(2)}, j_1^{(\s_1(2))})\delta(i_1^{(3)}, j_1^{(\s_1(3))}) \\
& \times \delta(i_2^{(1)}, j_2^{(\s_2(1))}) \delta(i_2^{(2)}, j_2^{(\s_2(2))})\\
& \times \delta(i_3^{(1)}, j_3^{(\s_3(1))}) \delta(i_3^{(2)}, j_2^{(\s_3(2))}) \delta(i_3^{(3)}, j_3^{(\s_3(3))})\,. 
\end{aligned}
\eea
\begin{figure}
\begin{center}
    \includegraphics[scale= 0.2]{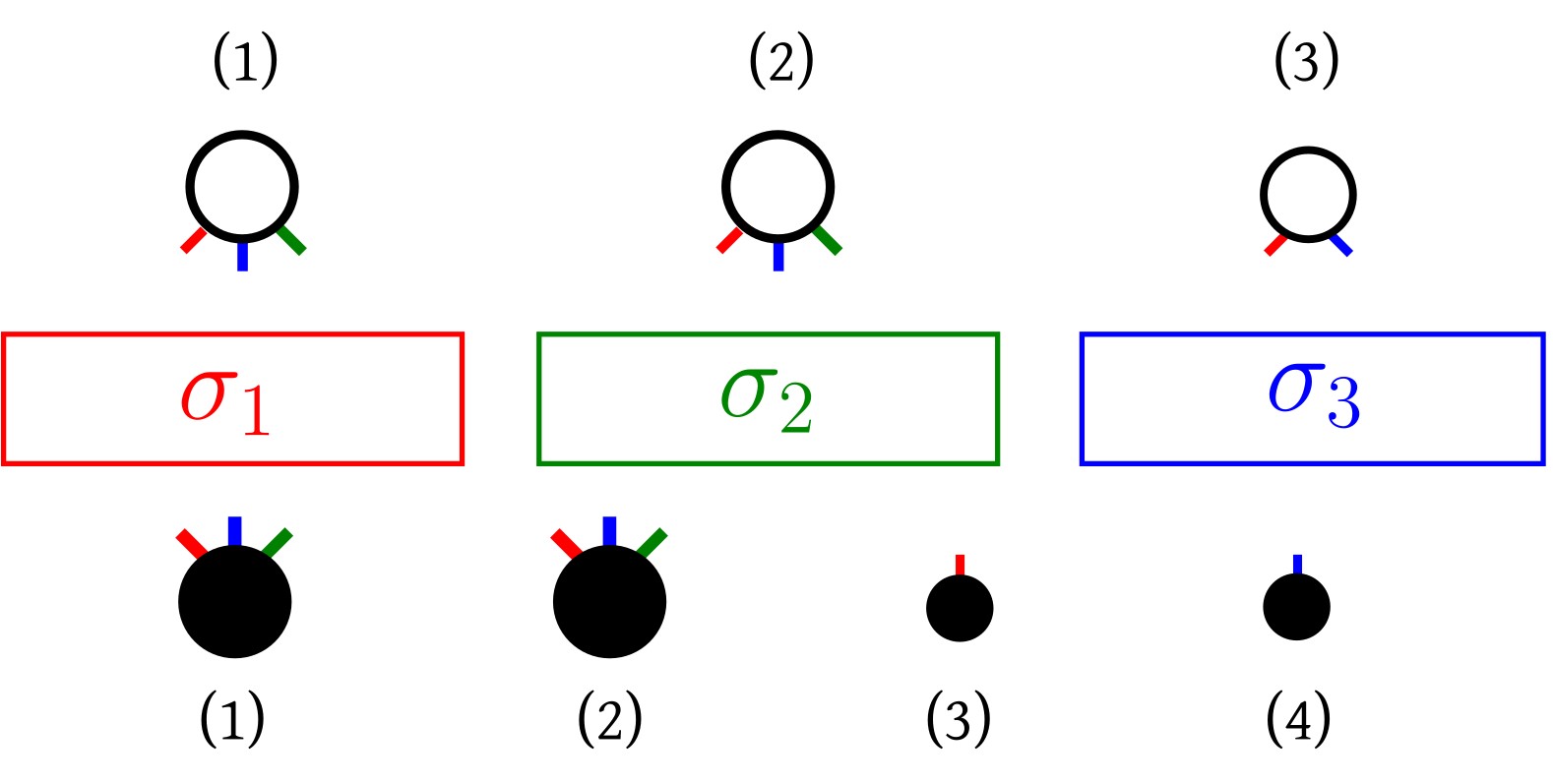}
\end{center}
\caption{Illustration of an edge-colored bipartite graph.}
\label{fig:IllusPregraph}
\end{figure}
Thanks to theorem \ref{theo:CountMultiOrder}, it is possible to count non-isomorphic contractions \textit{i.e.} non-isomorphic edge-colored bipartite graphs associated with $\Lambda$ and $\Gamma$. A Python implementation of the formula is given in appendix \ref{app:Code} with the code. We obtain 20 orbits of the group action. 

\subsection{Combinatorial proof}

The 20 orbits produced by the main counting scheme
can be recovered through combinatorial arguments.
The present section provides the proof
of this enumeration.

First, consider the case where the order-3 tensor part is disconnected from the matrix–vector part. It is known that a graph composed of two tensors $T$ and two tensors $\ov{T}$ of order 3 admits 4 non-isomorphic graphs (see table \ref{tab:formule_simple}, classical counting at fixed order). Since there is only one possible contraction for the matrix–vector part, this yields 4 non-isomorphic graphs in total.
These four graphs are shown in Figure~\ref{fig:1à4}.

\begin{figure}[ht]
\centering
\includegraphics[scale=0.37]{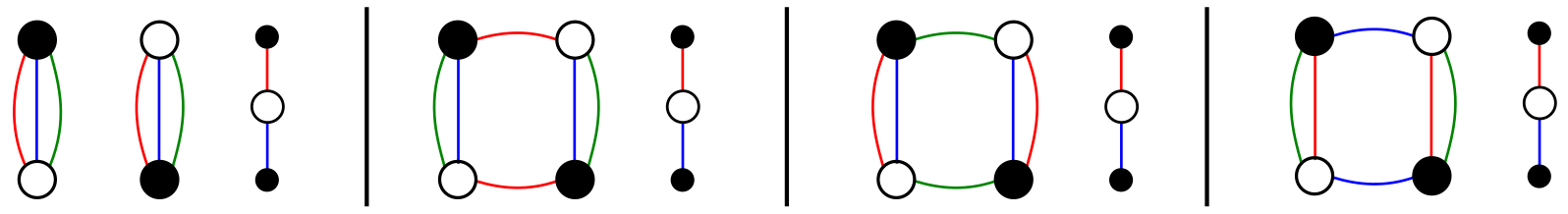}
\caption{Four graphs representing 
multiple-order contractions of tensors
issued from figure \ref{fig:IllusPregraph}: 
the order-3 tensors are not connected
to the vector and matrix fields.}
\label{fig:1à4}
\end{figure}

Starting from these 4 graphs, we construct the remaining 16 by perturbing the third-order tensor sector. To determine these contractions, we apply well-known graph-theoretic operations: cutting and gluing edges and half-edges. It is essential to ensure that the final graph remains edge-colored and bipartite, preserving 
the underlying unitary invariance of the tensor contraction. 

Select one of these graphs and choose one of the two red edges (which are symmetric). Cut the red edge, leaving two red half-edges on the vertices where it was previously incident. Next, consider the red half-edges of the vector and matrix fields, and glue them onto the two free half-edges in the tensor sector. From each graph in figure \ref{fig:1à4}, this process generates a new graph, depicted in figure \ref{fig:5à8}, corresponding to a new invariant.

\begin{figure}[ht]
\centering
\includegraphics[scale=0.37]{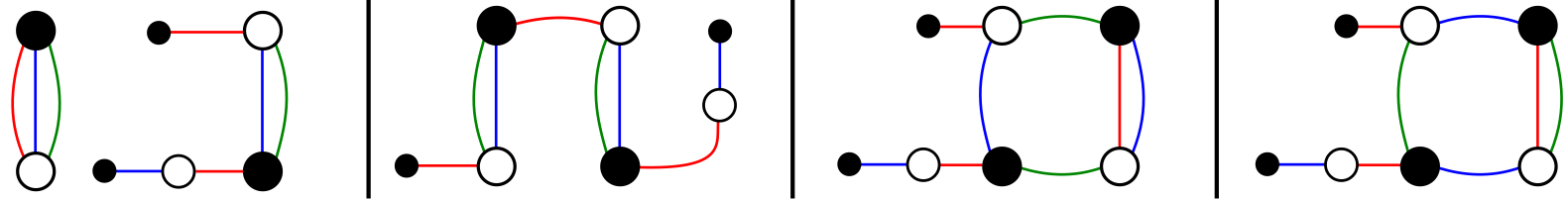}
\caption{Four graphs are obtained from figure \ref{fig:IllusPregraph} by gluing the red vector and matrix field to two different tensors with red half-edges.}
\label{fig:5à8}
\end{figure}

In a symmetric way, choosing the color blue
rather than the red in the above construction, we obtain 4 additional graphs. To generate the remaining eight configurations, we simultaneously contract the red and blue vectors within the third-order tensor sector. This requires cutting one red and one blue edges and then connecting the two vector fields. For each graph in Figure \ref{fig:1à4}, the vectors can be attached either to the same white vertex or to two different ones, resulting in exactly $2 \times 4$ distinct graphs, as shown in Figure \ref{fig:13à20}.

\begin{figure}[ht]
\centering
\includegraphics[scale=0.37]{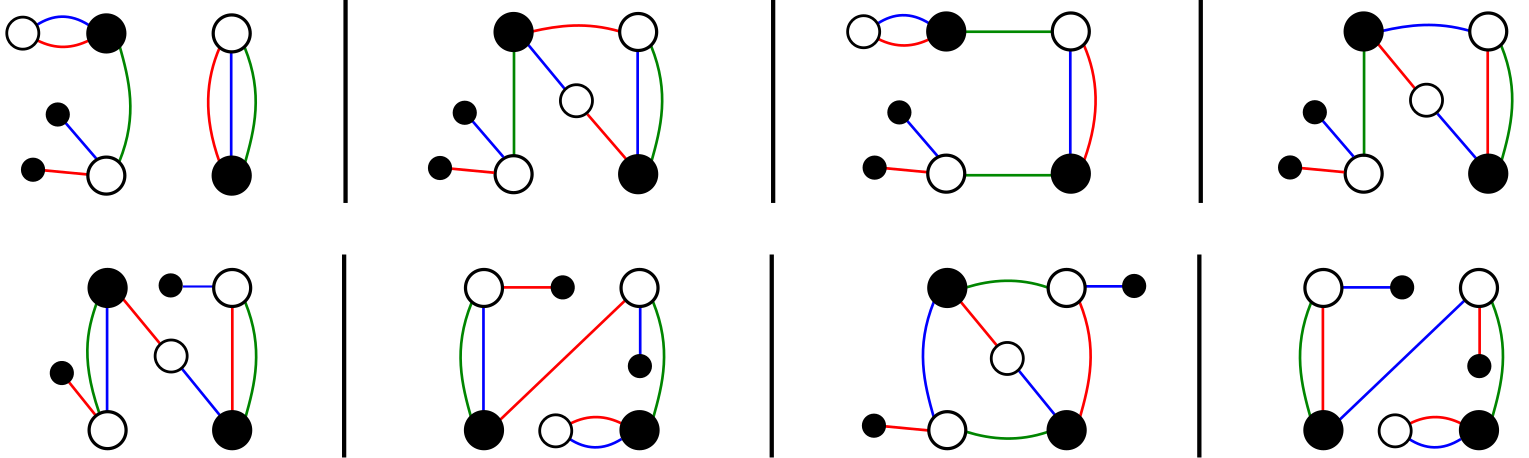}
\caption{Height configurations obtained from figure \ref{fig:IllusPregraph} by gluing the two vectors to one or two tensors.}
\label{fig:13à20}
\end{figure}

\section{Concluding remarks}
\label{sec:concl}

The definition of multiple-order tensor contraction is provided for any $d' \le d$, where $d$ is a positive integer, demonstrating that these tensor contractions remain invariant under the action of a tensor product of unitary groups and a (symmetric) group action on the set of complex tensors. By employing Burnside's lemma, a general counting formula is derived, extending the one introduced in previous work \cite{BenGeloun:2013lim} to a broader context. As an illustration, the formalism is applied to count invariants in a specific scenario and is complemented by a combinatorial approach to verify its consistency. It becomes evident that developing a combinatorial strategy for counting such invariants rapidly becomes intractable as the number of tensors increases, suggesting that computer-based counting methods may provide a unique
tractable way to determine all inequivalent configurations.

The perspectives of this work are numerous. Extending the Tensor Theory Space raises questions about the stability of these models under renormalization group flow. Assuming that vectors can be associated with matter degrees of freedom, a theory involving order $d$ tensors and vectors might resemble a $d$-dimensional gravity coupled with a vector matter field. However, this matter differs from that considered in Group Field Theory, where local degrees of freedom are associated with the tensor field content \cite{Oriti:2016ueo}. In any case, the perturbative renormalizability in such a Theory Space might differ significantly compared to that of tensor models with fixed order, starting with the Grosse-Wulkenhaar model \cite{Grosse_2005}, and even TFTs of fixed orders \cite{BenGeloun:2011rc, BenGeloun:2013vwi}. At the nonperturbative level, Functional Renormalization Group analyses should be conducted for this class of models to understand their critical behavior, which might reveal different properties compared to known TFTs with fixed order \cite{Benedetti:1411, BenGeloun:1508,BenGeloun:1601, BenGeloun:2016tmc}. 

As a concrete approach to the renormalization
group analysis of the multiple-order invariants, we present a coupled tensor-vector model generalizing the  
Grosse-Wulkenhaar model. Consider a complex order-3 tensor $(M_{mnp})_{m,n,p\in \N}$, and a vector $(\phi_k)_{k\in \N}$ of complex coordinates. 
An interesting candidate $(M^4\phi^4)$-like action is given by 
\bea
S[M,\phi] &=& 
\sum_{m,n,p} \ov{M}_{mnp}(m^{\alpha} + n^{\alpha} +p^{\alpha}+ \mu_1)M_{mnp}
+ \sum_{p} \ov{\phi}_{p}(p^{\beta}+ \mu_2)\phi_{p}
\crcr
&+&\lambda  \sum_{a,b,c,d,e,f,g,h} 
M_{{ a} { b} { e}}\ov{M}_{{ c} { b} { f}}M_{{ c} { d} { g}}\ov{M}_{{ a} { d} { h}} \phi_{ h} \ov{\phi}_{ e} \phi_{ f} \ov{\phi}_{ g}   
\eea
where $\alpha> 0$ and $\beta>0$ are real numbers, 
$\mu_1$ and $\mu_2$ play the role of mass couplings, and $\lambda$ the interaction coupling.  
The interaction is cyclic and pictured in 
figure \ref{fig:GrossEtWulkenhaar}. 
This action may lead to divergences for specific $\alpha$ and $\beta$. Establishing a power counting of these divergences would certainly be an interesting prospect.\begin{figure}[ht]
\centering
\includegraphics[scale=0.2]{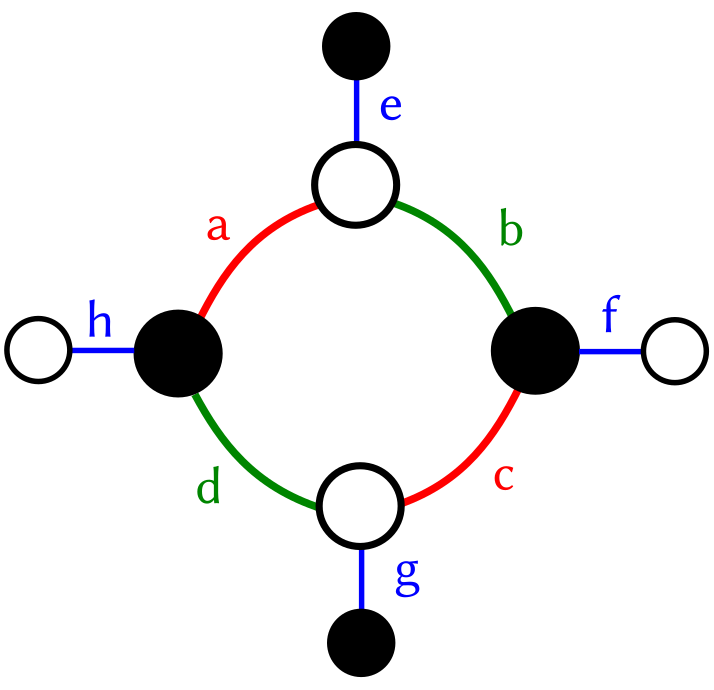}
\caption{A possible multiple-order
contraction generalizing the Grosse-Wulkenhaar trace interaction.}
\label{fig:GrossEtWulkenhaar}
\end{figure}

Another unexplored topic is the Topological Field Theory interpretation of the present counting formula. Could this formula count some covers of given topologies ? Generally, the connection is established \textit{via} Burnside's lemma. A closer examination of the proof of Theorem \ref{theo:CountMultiOrder} suggests that the required interpretation must be encoded   in \eqref{eq:ZBurnLemm}. Identifying the 2D complex associated with the Topological Field Theory linked to this counting remains crucial. While it might appear obvious from \eqref{eq:ZBurnLemm}, careful attention is needed as this action represents distinct subgroup actions across each component $\s_i$. Only a thorough analysis can provide an answer to this question.  
Another main theme in the domain consists in the computation of correlators in this framework and the identification of an effective representation basis for these correlators \cite{BenGeloun:2017vwn}. Once again, the representation theory of the symmetric group will serve as a central tool for addressing these questions.

\section*{Appendix}
\appendix
\renewcommand{\theproposition}{\Alph{section}.\arabic{proposition}}
\renewcommand{\theequation}{\Alph{section}.\arabic{equation}}
\renewcommand{\thesection}{\Alph{section}}
\setcounter{equation}{0} 
\setcounter{section}{0}  
\setcounter{proposition}{0}

\section{Invariance of fixed order contractions}
\label{app:Inv}
In this subsection, we provide the proof of proposition~\ref{prop:InvarianceContraction}. We then recall Burnside’s lemma and the orbit-stabilizer theorem, whose proofs can be found in any standard textbook on algebra, as they will be used repeatedly throughout the paper.

\begin{proposition*} 
Let $d,n \in \N^*$, $n$ tensors $T$ and $n$ tensors $\ov{T}$, both of order $d$ and each defined over a complex vector space $E$ of dimension $N$. Denote by $U(N)$ the unitary group, and consider the natural action of $U(N)^{\otimes d}$ on tensors of order $d$ via the fundamental representation on each index.

Let $\s \in S_n^{\times d}$ and
 let $I(\s; T)$ denote a specific index contraction between the tensors $T$ and $\ov{T}$ determined by $\s$.
 Then, 
\begin{itemize}
\item[(i)] $I(\s; T)$ is invariant under the action of $U(N)^{\otimes d}$ on the tensors $T$ and $\ov{T}$.
\item[(ii)]
 $\forall \g, \r \in S_n$, $I((\r\s_1\g, \cdots, \r\s_d\g); T) = I((\s_1, \cdots, \s_d); T) = I(\s; T)$.
\end{itemize}
\end{proposition*}

\begin{proof}
For $(i)$, we directly apply the rules (given by equation \eqref{eq:chtTenseurBase}) for the transformation of tensor coefficients for a unitary basis change matrix $U=[U_{ij}]_{i \le i,j \le N} \in U(N)$. We have:
\bea
\begin{aligned}
I(\s; T) &=  K(\s, \{i_1^{(j)}, \cdots, i_d^{(j)}\}_j, \{\bar i_1^{(k)} ,\cdots, i_d^{(k)}\}_k)\prod_{j = 1}^n T'_{i_1^{(j)} \cdots i_d^{(j)}} \prod_{k=1}^n \ov{T}'_{\bar i_1^{(k)} \cdots \bar i_d^{(k)}}\\
&= \prod_{j = 1}^n T'_{i_1^{(j)} \cdots i_d^{(j)}} \prod_{k=1}^n \ov{T}'_{ i_1^{\s_1(k)} \cdots i_d^{\s_d(k)}}.
\end{aligned}
\eea
In this second line, we have eliminated the contraction kernel by integrating all   Kronecker deltas. The indices $\bar{i}^{(j)}$ have thus disappeared, replaced by the indices $i^{\s(j)}$ to which they were linked. Now, applying the basis change, we obtain:
\bea
\begin{aligned}
I(\s; T) &= \prod_{j = 1}^n U_{i_1^{(j)} a_1^{(j)}}\cdots U_{i_d^{(j)} a_d^{(j)}}  T_{a_1^{(j)} \cdots a_d^{(j)}}\prod_{k=1}^n \ov{U}_{i_1^{\s_1(k)} b_1^{\s_1(k)}}\cdots\ov{U}_{i_d^{\s_d(k)} b_d^{\s_d(k)}} \ov{T}_{b_1^{\s_1(k)} \cdots b_d^{(\s_(k)}}\\
&=\prod_{l=1}^n U_{i_1^{(l)} a_1^{(l)}}\ov{U}_{i_1^{(l)} b_1^{(l)}} \cdots U_{i_d^{(l)} a_d^{(l)}}\ov{U}_{i_d^{(l)} b_d^{(l)}}\prod_{j = 1}^n T_{a_1^{(j)} \cdots a_d^{(j)}}\prod_{k=1}^n \ov{T}_{b_1^{\s_1(k)} \cdots b_d^{\s_d(k)}}.
\end{aligned}
\eea
In the last line, we paired the matrices because  when we change an index $i_c^{(l)}$ of a tensor $T$ by a coefficient $U_{i_c^{(l)}a_c^{(l)}}$, there is a tensor $\ov{T}$ that is contracted with the same index $i_c^{(l)}$, producing a coefficient $\ov{U}_{i_c^{(l)}b_c^{(l)}}$. These two coefficients are combined and must cancel using the fact that the matrix $U$ is unitary (\textit{i.e.}, $\ov{U}_{ij} = U^{-1}_{ji}$).  One gets: 
\bea
\begin{aligned}
I(\s; T')
&= \prod_{l=1}^n \delta(b_1^{(l)}, a_1^{(l)}) \cdots \delta(b_d^{(l)}, a_d^{(l)}) \prod_{j = 1}^n T_{a_1^{(j)} \cdots a_d^{(j)}}\prod_{k=1}^n \ov{T}_{b_1^{\s_1(k)} \cdots b_d^{\s_d(k)}}\\
&= \prod_{j = 1}^n T_{a_1^{(j)} \cdots a_d^{(j)}}\prod_{k=1}^n \ov{T}_{a_1^{\s_1(k)} \cdots a_d^{\s_d(k)}} \\
&= K(\s, \{i_1^{(j)}, \cdots, i_d^{(j)}\}_j, \{\bar i_1^{(k)} ,\cdots, i_d^{(k)}\}_k)\prod_{j = 1}^n T_{a_1^{(j)} \cdots a_d^{(j)}}\prod_{k=1}^n \ov{T}_{a_1^{(k)} \cdots a_d^{(k)}}\\
&= I(\s; T).
\end{aligned}
\eea
We now focus on $(ii)$. Let $\g, \r \in S_n$, then we claim that:
\bea K(\s, \{i_1^{(j)}, \cdots, i_d^{(j)}\}_j, \{\bar i_1^{(k)} ,\cdots, i_d^{(k)}\}_k) = \prod_{l=1}^n \prod_{c=1}^d \delta(i_c^{\s_c(l)}, \bar i_c^{(l)}) =  \prod_{l=1}^n \prod_{c=1}^d \delta(i_c^{\g\s_c(l)}, \bar i_c^{\g(l)}) \eea

We implement the fact that the numerical values of the contracted index exponents do not matter. What is important is that this value is unique and allows the contracted indices to be distinguished from each other. Therefore, these indices can be redefined by applying a permutation, say $\g$, to all of them. The following equality is obtained by making the variable substitution $l=\g^{-1}(l')$:
\bea
\begin{aligned}
\prod_{l=1}^n \prod_{c=1}^d \delta(i_c^{\g\s_c(l)}, \bar i_c^{\g(l)})
&= \prod_{l'=1}^n \prod_{c=1}^d \delta(i_c^{\g\s_c(\g^{-1}(l'))}, \bar i_c^{l'})\\
&=  K(\g\s\g^{-1}, \{i_1^{(j)}, \cdots, i_d^{(j)}\}_j, \{\bar i_1^{(k)} ,\cdots, i_d^{(k)}\}_k).
\end{aligned}
\eea
This gives us the equality: 
\bea \label{eq:demoinv1} 
I(\s;T) = I(\g\s\g^{-1};T).
\eea
On the other hand, when working with a contraction $I(\s;T)$, it is possible to change the order of the tensors in the product by making the variable substitution $k = \r(k')$.
\bea 
\label{eq:demoinv2}
I(\s;T) = \prod_{j = 1}^n T_{i_1^{(j)} \cdots i_d^{(j)}} \prod_{k'=1}^n \ov{T}_{ i_1^{\s_1(\r(k'))} \cdots i_d^{\s_d(\r(k'))}} = I(\s\r;T).
\eea
Combining equations \eqref{eq:demoinv1} and \eqref{eq:demoinv2} yields the 
desired point  $(ii)$.
\end{proof}

The main counting procedures of the unitary invariants intensively based 
on two  basic theorems in group theory. We recall them here. 

\begin{proposition}[Orbit-stabilizer theorem]
\label{prop:orbstab} 
Let $G$ be a finite group acting on a finite set X, et $x \in X$. We have
\bea| G | = |\Stab(x)||\Orb(x)|.\eea
\end{proposition}

\begin{proposition}[Burnside's lemma]
\label{prop:Burnside}
Let $G$ be a finite group acting over a finite set $X$. We denote $X^g = \left\{x \in X ~|~ g . x =x\right\}$ the set of fixed points of $g$. The number of orbits of the action of $G$ on $X$ is the average number of fixed points, \textit{i.e}.:

\bea |X/G| = \frac{1}{|G|} \sum_{g \in G} |X^g|.\eea
\end{proposition}

\section{Enumeration of unitary invariants}\label{sec:CalculOrbites}
In this section, we provide tables that, for given sets of parameters, enumerate the number of non-isomorphic edge-colored bipartite graphs, each corresponding to an inequivalent tensor contraction.

\subsection{Counting invariants at fixed order}
\label{subapp:tableauOrdreFixe}

In subsection \ref{sub:comptageInvMonoOrder},  theorem \ref{theo:JBandSajaye} establishes the  counting of  the number of possible non-isomorphic graphs representing the contractions of $n$ tensors $T$ and $n$ tensors $\ov{T}$, at fixed order $d$: $Z^d_n = \sum_{p ~ \vdash ~ n} \Sym(p)^{d-2}$,  where the sum runs over all integer partitions $p$ of $n$. The
following table \ref{tab:formule_simple} delivers the values of $Z^d_n$  for a range of $d$ and
and $n$. 

{\footnotesize{
\begin{table}[ht]
\begin{center}
\begin{tabular}{|c|c c c c c c c c |}
\hline
$n \setminus d$& $1$& $2$& $3$&$4$&$5$&$6$&$7$&$8$\\
\hline
$1$ & $1$& $1$& $1$        &$1$         &$1$         &$1$         &$1$         &$1$\\
$2$ & $1$& $2$& $4$        &$8$         &$16$        &$32$        &$64$        &$128$\\
$3$ & $1$& $3$& $11$       &$49$        &$251$       &$1393$      &$8051$      &$47449$\\
$4$ & $1$& $5$& $43$       &$681$       &$14491$     &$336465$    &$7997683$   &$1.91\e+08$\\
$5$ & $1$& $7$& $161$      &$14721$     &$1730861$   &$2.07\e+08$ &$2.49\e+10$ &$2.99\e+12$\\
$6$ & $1$& $11$& $901$     &$524137$    &$3.73\e+08$ &$2.69\e+11$ &$1.93\e+14$ &$1.39\e+17$\\
$7$ & $1$& $15$& $5579$    &$25471105$  &$1.28\e+11$ &$6.45\e+14$ &$3.25\e+18$ &$1.64\e+22$\\
$8$ & $1$& $22$& $43206$   &$1.63\e+09$ &$6.55\e+13$ &$2.64\e+18$ &$1.06\e+23$ &$4.30\e+27$\\
$9$ & $1$& $30$& $378360$  &$1.32\e+11$ &$4.78\e+16$ &$1.73\e+22$ &$6.29\e+27$ &$2.28\e+33$\\
$10$ & $1$& $42$&$3742738$ &$1.32\e+13$ &$4.77\e+19$ &$1.73\e+26$ &$6.29\e+32$ &$2.28\e+39$\\
\hline
\end{tabular}
\end{center}
\caption{Number of non-isomorphic graphs corresponding to the contractions of $n$ tensors $T$ and $n$ tensors $\ov{T}$ of order $d$.}
\label{tab:formule_simple}
\end{table} 
}}

\subsection{Counting multiple-order invariants}\label{subapp:tableauOrdreMultiple}

In subsection~\ref{subsec:ComptageInvMultiOrder}, we stated theorem~\ref{theo:CountMultiOrder}, which enumerates the number of non-isomorphic graphs arising from compatible colored sets of vertices $\Lambda$ and $\Gamma$ as defined previously. Below, we compute the number of such graphs, characterized by certain cardinality functions $n$ and $m$ corresponding to $\Lambda$ and $\Gamma$, respectively, as described in table~\ref{tab:classification}. Other choices are certainly possible, this is just an illustration. 

For $d \geq 3$, we consider $s$ tensors $T$ of order $d$ and $s-1$ tensors $\ov{R}$ of the same order. To enable the required contractions, we introduce additional matrices and vectors such that the configurations $\Lambda$ and $\Gamma$ remain compatible (at most 4 matrices of different types, and at most 2 vectors of different types). While the exact enumeration is feasible, we also provide   approximations for large parameter values. The results are presented in table~\ref{tab:multiple_order}.
Furthermore, applying the algorithm described in subsection~\ref{subapp:tableauOrdreFixe} to compute invariants of fixed order yields results that exactly match those presented in table~\ref{tab:formule_simple}.

{\footnotesize{
\begin{table}[ht]
\begin{center}
\begin{tabular}{|c|m{13cm}|}
\hline
order $d$ & Cardinality function $n$ (associated with $\Lambda$) and $m$ (associated with $\Gamma$); $s \in \N$\\
\hline 
$3$ & $n(\{1,2,3\}) = s, ~ m(\{1,2,3\}) = s-1, ~m(\{1,2\}) = 1, ~ m(\{3\}) = 1$ and $0$ for the other color types\\
\hline
$4$ & $n(\{1,2,3,4\}) = s, ~ m(\{1,2,3,4\}) = s-1, ~ m(\{1,2\}) = 1, ~ m(\{3\}) = 1, ~ m(\{4\})=1$ and $0$ for the other color types\\
\hline
$5$ & $n(\{1,2,3,4,5\}) = s, ~ m(\{1,2,3,4,5\}) = s-1, ~ m(\{1,2\}) = 1, ~ m(\{3,4\}) = 1, ~ m(\{5\})=1$ and $0$ for the other color types\\
\hline
$6$ & $n(\{1,2,3,4,5,6\}) = s, ~ m(\{1,2,3,4,5,6\}) = s-1, ~ m(\{1,2\}) = 1, ~ m(\{3,4\}) = 1, ~ m(\{5\})=1, ~ m(\{6\})=1$ and $0$ for the other color types \\
\hline
$7$ & $n(\{1,\cdots, 7\}) = s, ~ m(\{1,\cdots,7\}) = s-1, ~ m(\{1,2\}) = 1, ~ m(\{3,6\}) = 1, ~ m(\{5,7\})=1, ~ m(\{4\})=1$ and $0$ for the other color types\\
\hline
$8$ & $n(\{1,\cdots, 8\}) = s, ~ m(\{1,\cdots,8\}) = s-1, ~ m(\{1,2\}) = 1, ~ m(\{3,6\}) = 1, ~ m(\{5,7\})=1, ~ m(\{4,8\})=1$ and $0$ for the other color types \\
\hline
$9$ & $n(\{1,\cdots, 9\}) = s$ $m(\{1,\cdots,9\}) = s- 1, ~ m(\{1,2\}) = 1, ~ m(\{3,6\}) = 1, ~ m(\{5,7\})=1, ~ m(\{4,8\})=1, ~ m(\{9\})=1$ and $0$ for the other color types\\
\hline
\end{tabular}
\end{center}
\caption{Some compatible classifications for different orders.}
\label{tab:classification}
\end{table} }}

{\footnotesize{
\begin{table}[ht]
\begin{center}
\begin{tabular}{|c|c c c c c c c|}
\hline
$s \setminus d$& $3$&$4$&$5$&$6$&$7$&$8$&$9$\\
\hline
$1$ & $1$        &$1$         &$1$         &$1$         &$1$         &$1$&$1$\\
$2$ & $4$        &$8$         &$16$        &$32$        &$64$        &$128$&$256$\\
$3$ & $20$       &$112$        &$656$       &$3904$      &$23360$      &$140032$&$839936$\\
$4$ & $107$       &$2345$       &$55451$     &$1327697$    &$31852787$   &$7.64\e+08$&$1.83\e+10$\\
$5$ & $660$      &$72584$     &$8646192$   &$1.04\e+09$ &$1.24\e+11$ &$1.49\e+13$&$1.80\e+15$\\
$6$ & $4625$     &$3121289$    &$2.24\e+09$ &$1.61\e+12$ &$1.16\e+15$ &$8.36\e+17$&$6.02\e+20$\\
$7$ & $37108$    &$1.78\e+08$&  $8.96\e+11$ &$4.52\e+15$ &$2.28\e+19$ &$1.15\e+23$&$5.78\e+26$\\
$8$ & $334723$   &$1.30\e+10$ &$5.24\e+14$ &$2.11\e+19$ &$8.52\e+23$ &$3.44\e+28$&$1.40\e+33$\\
$9$ & $3359867$  &$1.19\e+12$ &$4.30\e+17$ &$1.56\e+23$ &$5.66\e+28$ &$2.05\e+34$&$7.46\e+39$\\
$10$ & $3.71\e+07$  &$1.32\e+14$ &$4.78\e+20$ &$1.73\e+27$ &$6.30\e+33$ &$2.28\e+40$&$8.29\e+46$\\
\hline
\end{tabular}
\end{center} 
\caption{Number of non-isomorphic graphs associated with the settings of table \ref{tab:classification}.}
\label{tab:multiple_order}
\end{table}
}}

We now define, for $d$ given by each row of table \ref{tab:multiple_order},
a sequence $(Z_s)^d_{s \in \N^*}$ that counts the number of non-isomorphic edge-colored graphs,
parametrized by $s$ black vertices (tensors of order $d$),
$s-1$ white vertices (tensors of order $d$), supplemented by at most 4 white vertices of valence 2 of different
color types (matrices), and at most 2 white leaves of different color
types (vectors). The choice made in the table
ensures that graph edge-coloring and bipartiteness are preserved.
We obtain at the first orders:
\bea
d=3&&\widetilde Z_{s}^3 = 1,4,20,107,660,4625,37108, \ldots  \\ 
d=4&&\widetilde Z_{s}^4 = 1,8,112,2345,72584,3121289 \ldots  \\ 
d=5&&\widetilde Z_{s}^5 = 1,16,656,55451,8646192,\ldots  
\eea
These sequences are not reported in the OEIS website.
Furthermore, we conjecture that, for several values
of the cardinality functions, several other sequences
are new.

\section{Codes}\label{app:Code}
In this section, we provide the Python code that was used to populate table \ref{tab:multiple_order} from table \ref{tab:classification} in appendix \ref{sec:CalculOrbites}.  This involves the implementation of the orbit counting formula given in theorem \ref{theo:CountMultiOrder}. Note that with the proposed solution, the calculations quickly take much more time as we consider classifications with more and more tensors.

\begin{verbatim}
class ColoredSetVertices: 
"""The construction epitomizes the notion of a colored set of vertices. 
The cardinality function is encoded, and it suffices for the enumeration 
of the orbits."""

    def __init__(self):
        self.__cardinal_function = dict()
        #{type: number of vertices associated to this type}
        self.__chromatic_index = 0 

    def __update_chromatic_index(self):
        """Every time the cardinality function is modified, 
        one shall check the new chromatic index."""
        maximums = []
        for type in self.__cardinal_function.keys():
            maximums.append(max(type))
        self.__chromatic_index = max(maximums)

    def set_colored_vertices(self, type: set, number: int):
        """Adds 'number' colored vertices of type 'type' """
        self.__cardinal_function.update({frozenset(type): number})
        if number != 0 : self.__update_chromatic_index()

    def set_cardinal_function(self, my_cardinal_function: dict):
        """Changes the old cardinallity function with the
        new one given in argument"""
        self.__cardinal_function = my_cardinal_function
        self.__update_chromatic_index()

    def get_cardinal_function(self):
        return self.__cardinal_function
    
    def get_chromatic_index(self):
        return self.__chromatic_index
\end{verbatim}
\begin{verbatim}
#----------------------------------
class PowerSet():
    """A class that encodes a power set {1,...,d} from a given  
    argument d (the order)."""

    def __init__(self, order: int):
        self.__order = order # = d
        
        init = {i for i in range(1, self.__order+1)}

        self.__P = set(chain.from_iterable({frozenset(e) for 
        e in combinations(init, r)} for r in range(self.__order+1)))
        #Creating the power set of {1,...,d} is drawn 
        #from the following web page:
        #https://towardsdatascience.com/the-subsets-powerset-of
        #-a-set-in-python-3-18e06bd85678/  

        self.__section = { c:set() for c in range(1, self.__order+1) }
        #self.__section[c] = {A \in P({1,..d}) | c  \in A}
        for s in self.__P:
            for c in s:
                self.__section[c].add(s) 

    def get_order(self):
        return self.__order
    
    def get_power_set(self):
        return self.__P
    
    def get_section(self, c: int):
        if 1 <= c and c <= self.__order: return self.__section[c]    
\end{verbatim}
\begin{verbatim}
#----------------------------------
def compatible(Lambda:ColoredSetVertices, Gamma:ColoredSetVertices
, power_set:PowerSet) -> bool:
    """Checks if two colored sets of vertices are compatible."""
    
    cf1 = Lambda.get_cardinal_function() #cardinality function n°1
    cf2 = Gamma.get_cardinal_function() #cardinality function n°2

    #Checks, for every c color, that the c-color 
    #multiplicities of Lambda and Gamma are equal. 
    for c in range(1,power_set.get_order()+1):
        n,m = 0,0
        for s in power_set.get_section(c):
            if s in cf1:
                n += cf1[s]
            if s in cf2:
                m += cf2[s]
        if n != m: return False

    return True

def convertIntegerPartitionToDict(p:list) -> dict:
    """Transforms a partition into a dict to highlight multiplicity."""
    p_dict = dict()
    for i in p:
        if i in p_dict: p_dict[i] += 1
        else: p_dict[i] = 1

    return p_dict

def createIntegerPartition(n:int) -> list:
    """The function 'ordered_partitions' from sympy module 
    delivers all integer partitions of n. Yet it gives the partition 
    as a list of all parts without taking into account multiplicities.
    Example: 5 = 1+1+1+1+1 => [1,1,1,1,1]
    or 6 = 3+2+1 => [3,2,1]
    We want to convert that list into a list of pairs (multiplicity, part)
    encoded in a dictionnary.
    Example: {1:5} => 1*5 = 5 
    or {1:, 2:1, 3:1} => 1*1 + 2*1 + 3*1 = 6
    or {1:2, 3:5, 8:1} => 1*2 + 3*5 + 8*1 = 25
    """
    return [convertIntegerPartitionToDict(p) for
            p in ordered_partitions(n)]

def symmetry_factor(partition:dict) -> int:
    """Computes the symmetry factor of an integer partition"""
    p = 1
    for i, m_i in partition.items():
        p*= i**m_i * factorial(m_i)
    return p

def calc_numerator(mu:dict, nu:dict, power_set:PowerSet, 
max_integer:int)->int:
    """The delta function condition consists 
    in looking at sums of partitions and check if they are equal."""
    chromatic_index = power_set.get_order()
    produit = 1
    for c in range(1, chromatic_index+1):
        p1 = {i:0 for i in range(1, max_integer+1)}
        p2 = {i:0 for i in range(1, max_integer+1)}

        for s in power_set.get_section(c):
            if s in mu:
                for n,multiplicty in mu[s].items():
                    p1[n] += multiplicty
            if s in nu:
                for n,multiplicty in nu[s].items():
                    p2[n] += multiplicty
        if p1 != p2: return 0 #delta condition
        else: produit *= symmetry_factor(p1)
        
    return produit

def calc_denominator(mu: dict, nu: dict, power_set) -> int:
    produit = 1
    for s in power_set.get_power_set():
        if s in mu: produit *= symmetry_factor(mu[s])
        if s in nu: produit *= symmetry_factor(nu[s])
    return produit

def couting_orbits(Lambda:ColoredSetVertices,
Gamma:ColoredSetVertices) -> float:
    chromatic_index = max(Lambda.get_chromatic_index(), 
    Gamma.get_chromatic_index())
    power_set = PowerSet(chromatic_index)
    max_integer = max(max(Lambda.get_cardinal_function().values()), 
    max(Gamma.get_cardinal_function().values()))
    #This is the larger number that will be partitioned
    
    if compatible(Lambda, Gamma, power_set):
        cf1 = Lambda.get_cardinal_function() #cardinality function n°1
        cf2 = Gamma.get_cardinal_function() #cardinality function n°2

        integer_partition_cf1 = {frozenset(s): 
        createIntegerPartition(cf1[s]) for s in cf1.keys()}
        integer_partition_cf2 = {frozenset(s):
        createIntegerPartition(cf2[s]) for s in cf2.keys()}
        #Integer partitions associated with the cardinality functions of 
        #Lambda and Gamma
        somme  = 0
        #This part encodes the sum over
        #all partitions over the cardinality functions
        for x in itertools.product(*integer_partition_cf1.values()):
            mu = dict(zip(cf1.keys(), x))
            for y in itertools.product(*integer_partition_cf2.values()):
                nu = dict(zip(cf2.keys(), y))
                num = calc_numerator(mu, nu, power_set, max_integer)
                
                if num:
                    den = calc_denominator(mu, nu, power_set)
                    somme += num/den
        
        return somme
    else:
        print("The two colored sets of vertices are not compatible.")
        
    return 0.    
\end{verbatim}

\bibliographystyle{apsrev4-2-titles}
\bibliography{mybib}

\begin{thebibliography}{54}%
\makeatletter
\providecommand \@ifxundefined [1]{%
 \@ifx{#1\undefined}
}%
\providecommand \@ifnum [1]{%
 \ifnum #1\expandafter \@firstoftwo
 \else \expandafter \@secondoftwo
 \fi
}%
\providecommand \@ifx [1]{%
 \ifx #1\expandafter \@firstoftwo
 \else \expandafter \@secondoftwo
 \fi
}%
\providecommand \natexlab [1]{#1}%
\providecommand \enquote  [1]{``#1''}%
\providecommand \bibnamefont  [1]{#1}%
\providecommand \bibfnamefont [1]{#1}%
\providecommand \citenamefont [1]{#1}%
\providecommand \href@noop [0]{\@secondoftwo}%
\providecommand \href [0]{\begingroup \@sanitize@url \@href}%
\providecommand \@href[1]{\@@startlink{#1}\@@href}%
\providecommand \@@href[1]{\endgroup#1\@@endlink}%
\providecommand \@sanitize@url [0]{\catcode `\\12\catcode `\$12\catcode
  `\&12\catcode `\#12\catcode `\^12\catcode `\_12\catcode `\%12\relax}%
\providecommand \@@startlink[1]{}%
\providecommand \@@endlink[0]{}%
\providecommand \url  [0]{\begingroup\@sanitize@url \@url }%
\providecommand \@url [1]{\endgroup\@href {#1}{\urlprefix }}%
\providecommand \urlprefix  [0]{URL }%
\providecommand \Eprint [0]{\href }%
\providecommand \doibase [0]{https://doi.org/}%
\providecommand \selectlanguage [0]{\@gobble}%
\providecommand \bibinfo  [0]{\@secondoftwo}%
\providecommand \bibfield  [0]{\@secondoftwo}%
\providecommand \translation [1]{[#1]}%
\providecommand \BibitemOpen [0]{}%
\providecommand \bibitemStop [0]{}%
\providecommand \bibitemNoStop [0]{.\EOS\space}%
\providecommand \EOS [0]{\spacefactor3000\relax}%
\providecommand \BibitemShut  [1]{\csname bibitem#1\endcsname}%
\let\auto@bib@innerbib\@empty
\bibitem [{\citenamefont {Di~Francesco}\ \emph {et~al.}(1995)\citenamefont
  {Di~Francesco}, \citenamefont {Ginsparg},\ and\ \citenamefont
  {Zinn-Justin}}]{DiFrancesco:1993cyw}%
  \BibitemOpen
  \bibfield  {author} {\bibinfo {author} {\bibfnamefont {P.}~\bibnamefont
  {Di~Francesco}}, \bibinfo {author} {\bibfnamefont {P.~H.}\ \bibnamefont
  {Ginsparg}},\ and\ \bibinfo {author} {\bibfnamefont {J.}~\bibnamefont
  {Zinn-Justin}},\ }\bibfield  {title} {\emph {\bibinfo {title} {{2-D Gravity
  and random matrices}}},\ }\href
  {https://doi.org/10.1016/0370-1573(94)00084-G} {\bibfield  {journal}
  {\bibinfo  {journal} {Phys. Rept.}\ }\textbf {\bibinfo {volume} {254}},\
  \bibinfo {pages} {1--133} (\bibinfo {year} {1995})},\ \Eprint
  {https://arxiv.org/abs/hep-th/9306153} {arXiv:hep-th/9306153} \BibitemShut
  {NoStop}%
\bibitem [{\citenamefont {Le~Gall}\ and\ \citenamefont
  {Miermont}(2010)}]{LE_GALL_2010}%
  \BibitemOpen
  \bibfield  {author} {\bibinfo {author} {\bibfnamefont {J.-F.}\ \bibnamefont
  {Le~Gall}}\ and\ \bibinfo {author} {\bibfnamefont {G.}~\bibnamefont
  {Miermont}},\ }in\ \href {https://doi.org/10.1142/9789814304634_0037} {\emph
  {\bibinfo {booktitle} {XVIth International Congress on Mathematical
  Physics}}}\ (\bibinfo  {publisher} {World Scientific},\ \bibinfo {year}
  {2010})\ \bibinfo {note} {arXiv:0907.3262 [math.PR]}\BibitemShut {NoStop}%
\bibitem [{\citenamefont {Ambjorn}\ \emph {et~al.}(1991)\citenamefont
  {Ambjorn}, \citenamefont {Durhuus},\ and\ \citenamefont
  {Jonsson}}]{Ambjorn:1990ge}%
  \BibitemOpen
  \bibfield  {author} {\bibinfo {author} {\bibfnamefont {J.}~\bibnamefont
  {Ambjorn}}, \bibinfo {author} {\bibfnamefont {B.}~\bibnamefont {Durhuus}},\
  and\ \bibinfo {author} {\bibfnamefont {T.}~\bibnamefont {Jonsson}},\
  }\bibfield  {title} {\emph {\bibinfo {title} {{Three-dimensional simplicial
  quantum gravity and generalized matrix models}}},\ }\href
  {https://doi.org/10.1142/S0217732391001184} {\bibfield  {journal} {\bibinfo
  {journal} {Mod. Phys. Lett. A}\ }\textbf {\bibinfo {volume} {6}},\ \bibinfo
  {pages} {1133--1146} (\bibinfo {year} {1991})}\BibitemShut {NoStop}%
\bibitem [{\citenamefont {Sasakura}(1991)}]{Sasakura:1990fs}%
  \BibitemOpen
  \bibfield  {author} {\bibinfo {author} {\bibfnamefont {N.}~\bibnamefont
  {Sasakura}},\ }\bibfield  {title} {\emph {\bibinfo {title} {{Tensor model for
  gravity and orientability of manifold}}},\ }\href
  {https://doi.org/10.1142/S0217732391003055} {\bibfield  {journal} {\bibinfo
  {journal} {Mod. Phys. Lett. A}\ }\textbf {\bibinfo {volume} {6}},\ \bibinfo
  {pages} {2613--2624} (\bibinfo {year} {1991})}\BibitemShut {NoStop}%
\bibitem [{\citenamefont {Gross}(1992)}]{Gross:1991hx}%
  \BibitemOpen
  \bibfield  {author} {\bibinfo {author} {\bibfnamefont {M.}~\bibnamefont
  {Gross}},\ }\bibfield  {title} {\emph {\bibinfo {title} {{Tensor models and
  simplicial quantum gravity in \ensuremath{>} 2-D}}},\ }\href
  {https://doi.org/10.1016/S0920-5632(05)80015-5} {\bibfield  {journal}
  {\bibinfo  {journal} {Nucl. Phys. B Proc. Suppl.}\ }\textbf {\bibinfo
  {volume} {25}},\ \bibinfo {pages} {144--149} (\bibinfo {year}
  {1992})}\BibitemShut {NoStop}%
\bibitem [{\citenamefont {Boulatov}(1992)}]{Boulatov:1992vp}%
  \BibitemOpen
  \bibfield  {author} {\bibinfo {author} {\bibfnamefont {D.~V.}\ \bibnamefont
  {Boulatov}},\ }\bibfield  {title} {\emph {\bibinfo {title} {{A Model of
  three-dimensional lattice gravity}}},\ }\href
  {https://doi.org/10.1142/S0217732392001324} {\bibfield  {journal} {\bibinfo
  {journal} {Mod. Phys. Lett. A}\ }\textbf {\bibinfo {volume} {7}},\ \bibinfo
  {pages} {1629--1646} (\bibinfo {year} {1992})},\ \Eprint
  {https://arxiv.org/abs/hep-th/9202074} {arXiv:hep-th/9202074} \BibitemShut
  {NoStop}%
\bibitem [{\citenamefont {Freidel}(2005)}]{Freidel_2005}%
  \BibitemOpen
  \bibfield  {author} {\bibinfo {author} {\bibfnamefont {L.}~\bibnamefont
  {Freidel}},\ }\bibfield  {title} {\emph {\bibinfo {title} {Group field
  theory: An overview}},\ }\href {https://doi.org/10.1007/s10773-005-8894-1}
  {\bibfield  {journal} {\bibinfo  {journal} {International Journal of
  Theoretical Physics}\ }\textbf {\bibinfo {volume} {44}},\ \bibinfo {pages}
  {1769–1783} (\bibinfo {year} {2005})}\BibitemShut {NoStop}%
\bibitem [{\citenamefont {Oriti}(2007)}]{Oriti:2006se}%
  \BibitemOpen
  \bibfield  {author} {\bibinfo {author} {\bibfnamefont {D.}~\bibnamefont
  {Oriti}},\ }\href {https://arxiv.org/abs/gr-qc/0607032} {\bibinfo {title}
  {The group field theory approach to quantum gravity}} (\bibinfo {year}
  {2007}),\ \Eprint {https://arxiv.org/abs/gr-qc/0607032} {arXiv:gr-qc/0607032
  [gr-qc]} \BibitemShut {NoStop}%
\bibitem [{\citenamefont {Marchetti}\ \emph {et~al.}(2023)\citenamefont
  {Marchetti}, \citenamefont {Oriti}, \citenamefont {Pithis},\ and\
  \citenamefont {Th\"urigen}}]{Marchetti:2022nrf}%
  \BibitemOpen
  \bibfield  {author} {\bibinfo {author} {\bibfnamefont {L.}~\bibnamefont
  {Marchetti}}, \bibinfo {author} {\bibfnamefont {D.}~\bibnamefont {Oriti}},
  \bibinfo {author} {\bibfnamefont {A.~G.~A.}\ \bibnamefont {Pithis}},\ and\
  \bibinfo {author} {\bibfnamefont {J.}~\bibnamefont {Th\"urigen}},\ }\bibfield
   {title} {\emph {\bibinfo {title} {{Mean-Field Phase Transitions in Tensorial
  Group Field Theory Quantum Gravity}}},\ }\href
  {https://doi.org/10.1103/PhysRevLett.130.141501} {\bibfield  {journal}
  {\bibinfo  {journal} {Phys. Rev. Lett.}\ }\textbf {\bibinfo {volume} {130}},\
  \bibinfo {pages} {141501} (\bibinfo {year} {2023})},\ \Eprint
  {https://arxiv.org/abs/2211.12768} {arXiv:2211.12768 [gr-qc]} \BibitemShut
  {NoStop}%
\bibitem [{\citenamefont {Carrozza}(2024)}]{Carrozza:2024gnh}%
  \BibitemOpen
  \bibfield  {author} {\bibinfo {author} {\bibfnamefont {S.}~\bibnamefont
  {Carrozza}},\ }\href@noop {} {\bibinfo {title} {{Tensor models and group
  field theories: combinatorics, large $N$ and renormalization}}} (\bibinfo
  {year} {2024}),\ \Eprint {https://arxiv.org/abs/2404.07834} {arXiv:2404.07834
  [math-ph]} \BibitemShut {NoStop}%
\bibitem [{\citenamefont {Carrozza}(2016)}]{carrozza2016tensorial}%
  \BibitemOpen
  \bibfield  {author} {\bibinfo {author} {\bibfnamefont {S.}~\bibnamefont
  {Carrozza}},\ }\href {https://books.google.co.jp/books?id=YmFjvgAACAAJ}
  {\emph {\bibinfo {title} {Tensorial Methods and Renormalization in Group
  Field Theories}}},\ Springer Theses\ (\bibinfo  {publisher} {Springer
  International Publishing},\ \bibinfo {year} {2016})\BibitemShut {NoStop}%
\bibitem [{\citenamefont {Rivasseau}(2012)}]{Rivasseau:2011hm}%
  \BibitemOpen
  \bibfield  {author} {\bibinfo {author} {\bibfnamefont {V.}~\bibnamefont
  {Rivasseau}},\ }\bibfield  {title} {\emph {\bibinfo {title} {{Quantum Gravity
  and Renormalization: The Tensor Track}}},\ }\href
  {https://doi.org/10.1063/1.4715396} {\bibfield  {journal} {\bibinfo
  {journal} {AIP Conf. Proc.}\ }\textbf {\bibinfo {volume} {1444}},\ \bibinfo
  {pages} {18--29} (\bibinfo {year} {2012})},\ \Eprint
  {https://arxiv.org/abs/1112.5104} {arXiv:1112.5104 [hep-th]} \BibitemShut
  {NoStop}%
\bibitem [{\citenamefont {Gurau}(2017)}]{gurau2017random}%
  \BibitemOpen
  \bibfield  {author} {\bibinfo {author} {\bibfnamefont {R.}~\bibnamefont
  {Gurau}},\ }\href {https://books.google.fr/books?id=6RjGDQAAQBAJ} {\emph
  {\bibinfo {title} {Random Tensors}}}\ (\bibinfo  {publisher} {Oxford
  University Press},\ \bibinfo {year} {2017})\BibitemShut {NoStop}%
\bibitem [{\citenamefont {Bonzom}\ \emph {et~al.}(2011)\citenamefont {Bonzom},
  \citenamefont {Gurau}, \citenamefont {Riello},\ and\ \citenamefont
  {Rivasseau}}]{Bonzom_2011}%
  \BibitemOpen
  \bibfield  {author} {\bibinfo {author} {\bibfnamefont {V.}~\bibnamefont
  {Bonzom}}, \bibinfo {author} {\bibfnamefont {R.}~\bibnamefont {Gurau}},
  \bibinfo {author} {\bibfnamefont {A.}~\bibnamefont {Riello}},\ and\ \bibinfo
  {author} {\bibfnamefont {V.}~\bibnamefont {Rivasseau}},\ }\bibfield  {title}
  {\emph {\bibinfo {title} {Critical behavior of colored tensor models in the
  large n limit}},\ }\href {https://doi.org/10.1016/j.nuclphysb.2011.07.022}
  {\bibfield  {journal} {\bibinfo  {journal} {Nuclear Physics B}\ }\textbf
  {\bibinfo {volume} {853}},\ \bibinfo {pages} {174–195} (\bibinfo {year}
  {2011})}\BibitemShut {NoStop}%
\bibitem [{\citenamefont {Gurau}\ and\ \citenamefont
  {Ryan}(2013)}]{Gurau_2013}%
  \BibitemOpen
  \bibfield  {author} {\bibinfo {author} {\bibfnamefont {R.}~\bibnamefont
  {Gurau}}\ and\ \bibinfo {author} {\bibfnamefont {J.~P.}\ \bibnamefont
  {Ryan}},\ }\bibfield  {title} {\emph {\bibinfo {title} {Melons are branched
  polymers}},\ }\href {https://doi.org/10.1007/s00023-013-0291-3} {\bibfield
  {journal} {\bibinfo  {journal} {Annales Henri Poincaré}\ }\textbf {\bibinfo
  {volume} {15}},\ \bibinfo {pages} {2085–2131} (\bibinfo {year}
  {2013})}\BibitemShut {NoStop}%
\bibitem [{\citenamefont {Ben~Geloun}\ and\ \citenamefont
  {Rivasseau}(2013)}]{BenGeloun:2011rc}%
  \BibitemOpen
  \bibfield  {author} {\bibinfo {author} {\bibfnamefont {J.}~\bibnamefont
  {Ben~Geloun}}\ and\ \bibinfo {author} {\bibfnamefont {V.}~\bibnamefont
  {Rivasseau}},\ }\bibfield  {title} {\emph {\bibinfo {title} {{A
  Renormalizable 4-Dimensional Tensor Field Theory}}},\ }\href
  {https://doi.org/10.1007/s00220-012-1549-1} {\bibfield  {journal} {\bibinfo
  {journal} {Commun. Math. Phys.}\ }\textbf {\bibinfo {volume} {318}},\
  \bibinfo {pages} {69--109} (\bibinfo {year} {2013})},\ \Eprint
  {https://arxiv.org/abs/1111.4997} {arXiv:1111.4997 [hep-th]} \BibitemShut
  {NoStop}%
\bibitem [{\citenamefont {Ben~Geloun}(2014)}]{BenGeloun:2013vwi}%
  \BibitemOpen
  \bibfield  {author} {\bibinfo {author} {\bibfnamefont {J.}~\bibnamefont
  {Ben~Geloun}},\ }\bibfield  {title} {\emph {\bibinfo {title} {{Renormalizable
  Models in Rank $d\geq 2$ Tensorial Group Field Theory}}},\ }\href
  {https://doi.org/10.1007/s00220-014-2142-6} {\bibfield  {journal} {\bibinfo
  {journal} {Commun. Math. Phys.}\ }\textbf {\bibinfo {volume} {332}},\
  \bibinfo {pages} {117--188} (\bibinfo {year} {2014})},\ \Eprint
  {https://arxiv.org/abs/1306.1201} {arXiv:1306.1201 [hep-th]} \BibitemShut
  {NoStop}%
\bibitem [{\citenamefont {Eichhorn}\ and\ \citenamefont
  {Koslowski}(2013)}]{Eichhorn:1309}%
  \BibitemOpen
  \bibfield  {author} {\bibinfo {author} {\bibfnamefont {A.}~\bibnamefont
  {Eichhorn}}\ and\ \bibinfo {author} {\bibfnamefont {T.~A.}\ \bibnamefont
  {Koslowski}},\ }\bibfield  {title} {\emph {\bibinfo {title} {{Continuum limit
  in matrix models for quantum gravity from the Functional Renormalization
  Group}}},\ }\href {https://doi.org/10.1103/PhysRevD.88.084016} {\bibfield
  {journal} {\bibinfo  {journal} {Phys. Rev. D}\ }\textbf {\bibinfo {volume}
  {88}},\ \bibinfo {pages} {084016} (\bibinfo {year} {2013})},\ \Eprint
  {https://arxiv.org/abs/1309.1690} {1309.1690} \BibitemShut {NoStop}%
\bibitem [{\citenamefont {Eichhorn}\ and\ \citenamefont
  {Koslowski}(2014)}]{Eichhorn:1408}%
  \BibitemOpen
  \bibfield  {author} {\bibinfo {author} {\bibfnamefont {A.}~\bibnamefont
  {Eichhorn}}\ and\ \bibinfo {author} {\bibfnamefont {T.~A.}\ \bibnamefont
  {Koslowski}},\ }\bibfield  {title} {\emph {\bibinfo {title} {{Towards phase
  transitions between discrete and continuum quantum spacetime from the
  renormalization group}}},\ }\href
  {https://doi.org/10.1103/PhysRevD.90.104039} {\bibfield  {journal} {\bibinfo
  {journal} {Phys. Rev. D}\ }\textbf {\bibinfo {volume} {90}},\ \bibinfo
  {pages} {104039} (\bibinfo {year} {2014})},\ \Eprint
  {https://arxiv.org/abs/1408.4127} {1408.4127} \BibitemShut {NoStop}%
\bibitem [{\citenamefont {Benedetti}\ \emph {et~al.}(2015)\citenamefont
  {Benedetti}, \citenamefont {Ben~Geloun},\ and\ \citenamefont
  {Oriti}}]{Benedetti:1411}%
  \BibitemOpen
  \bibfield  {author} {\bibinfo {author} {\bibfnamefont {D.}~\bibnamefont
  {Benedetti}}, \bibinfo {author} {\bibfnamefont {J.}~\bibnamefont
  {Ben~Geloun}},\ and\ \bibinfo {author} {\bibfnamefont {D.}~\bibnamefont
  {Oriti}},\ }\bibfield  {title} {\emph {\bibinfo {title} {{Functional
  Renormalisation Group Approach for Tensorial Group Field Theory: a Rank-3
  Model}}},\ }\href {https://doi.org/10.1007/JHEP03(2015)084} {\bibfield
  {journal} {\bibinfo  {journal} {JHEP}\ }\textbf {\bibinfo {volume} {03}},\
  \bibinfo {pages} {084}},\ \Eprint {https://arxiv.org/abs/1411.3180}
  {1411.3180} \BibitemShut {NoStop}%
\bibitem [{\citenamefont {Ben~Geloun}\ \emph {et~al.}(2015)\citenamefont
  {Ben~Geloun}, \citenamefont {Martini},\ and\ \citenamefont
  {Oriti}}]{BenGeloun:1508}%
  \BibitemOpen
  \bibfield  {author} {\bibinfo {author} {\bibfnamefont {J.}~\bibnamefont
  {Ben~Geloun}}, \bibinfo {author} {\bibfnamefont {R.}~\bibnamefont
  {Martini}},\ and\ \bibinfo {author} {\bibfnamefont {D.}~\bibnamefont
  {Oriti}},\ }\bibfield  {title} {\emph {\bibinfo {title} {{Functional
  Renormalization Group analysis of a Tensorial Group Field Theory on
  ${R}^3$}}},\ }\href {https://doi.org/10.1209/0295-5075/112/31001} {\bibfield
  {journal} {\bibinfo  {journal} {EPL}\ }\textbf {\bibinfo {volume} {112}},\
  \bibinfo {pages} {31001} (\bibinfo {year} {2015})},\ \Eprint
  {https://arxiv.org/abs/1508.01855} {1508.01855} \BibitemShut {NoStop}%
\bibitem [{\citenamefont {Ben~Geloun}\ \emph {et~al.}(2016)\citenamefont
  {Ben~Geloun}, \citenamefont {Martini},\ and\ \citenamefont
  {Oriti}}]{BenGeloun:1601}%
  \BibitemOpen
  \bibfield  {author} {\bibinfo {author} {\bibfnamefont {J.}~\bibnamefont
  {Ben~Geloun}}, \bibinfo {author} {\bibfnamefont {R.}~\bibnamefont
  {Martini}},\ and\ \bibinfo {author} {\bibfnamefont {D.}~\bibnamefont
  {Oriti}},\ }\bibfield  {title} {\emph {\bibinfo {title} {{Functional
  renormalization group analysis of tensorial group field theories on
  $R^d$}}},\ }\href {https://doi.org/10.1103/PhysRevD.94.024017} {\bibfield
  {journal} {\bibinfo  {journal} {Phys. Rev. D}\ }\textbf {\bibinfo {volume}
  {94}},\ \bibinfo {pages} {024017} (\bibinfo {year} {2016})},\ \Eprint
  {https://arxiv.org/abs/1601.08211} {1601.08211} \BibitemShut {NoStop}%
\bibitem [{\citenamefont {Krajewski}\ and\ \citenamefont
  {Toriumi}(2016)}]{Krajewski:2016jf}%
  \BibitemOpen
  \bibfield  {author} {\bibinfo {author} {\bibfnamefont {T.}~\bibnamefont
  {Krajewski}}\ and\ \bibinfo {author} {\bibfnamefont {R.}~\bibnamefont
  {Toriumi}},\ }\bibfield  {title} {\emph {\bibinfo {title} {{Exact
  Renormalisation Group Equations and Loop Equations for Tensor Models}}},\
  }\href@noop {} {\bibfield  {journal} {\bibinfo  {journal} {SIGMA}\ }\textbf
  {\bibinfo {volume} {12}},\ \bibinfo {pages} {068} (\bibinfo {year} {2016})},\
  \Eprint {https://arxiv.org/abs/1603.00172} {1603.00172} \BibitemShut
  {NoStop}%
\bibitem [{\citenamefont {Eichhorn}\ \emph
  {et~al.}(2019{\natexlab{a}})\citenamefont {Eichhorn}, \citenamefont {Lumma},
  \citenamefont {Koslowski},\ and\ \citenamefont {Pereira}}]{Eichhorn:1811}%
  \BibitemOpen
  \bibfield  {author} {\bibinfo {author} {\bibfnamefont {A.}~\bibnamefont
  {Eichhorn}}, \bibinfo {author} {\bibfnamefont {J.}~\bibnamefont {Lumma}},
  \bibinfo {author} {\bibfnamefont {T.~A.}\ \bibnamefont {Koslowski}},\ and\
  \bibinfo {author} {\bibfnamefont {A.~D.}\ \bibnamefont {Pereira}},\
  }\bibfield  {title} {\emph {\bibinfo {title} {{Towards background independent
  quantum gravity with tensor models}}},\ }\href
  {https://doi.org/10.1088/1361-6382/ab2545} {\bibfield  {journal} {\bibinfo
  {journal} {Class. Quant. Grav.}\ }\textbf {\bibinfo {volume} {36}},\ \bibinfo
  {pages} {155007} (\bibinfo {year} {2019}{\natexlab{a}})},\ \Eprint
  {https://arxiv.org/abs/1811.00814} {1811.00814} \BibitemShut {NoStop}%
\bibitem [{\citenamefont {Eichhorn}\ \emph
  {et~al.}(2019{\natexlab{b}})\citenamefont {Eichhorn}, \citenamefont
  {Koslowski}, \citenamefont {Lumma},\ and\ \citenamefont
  {Pereira}}]{Eichhorn:2018ylk}%
  \BibitemOpen
  \bibfield  {author} {\bibinfo {author} {\bibfnamefont {A.}~\bibnamefont
  {Eichhorn}}, \bibinfo {author} {\bibfnamefont {T.}~\bibnamefont {Koslowski}},
  \bibinfo {author} {\bibfnamefont {J.}~\bibnamefont {Lumma}},\ and\ \bibinfo
  {author} {\bibfnamefont {A.~D.}\ \bibnamefont {Pereira}},\ }\bibfield
  {title} {\emph {\bibinfo {title} {{Towards background independent quantum
  gravity with tensor models}}},\ }\href
  {https://doi.org/10.1088/1361-6382/ab2545} {\bibfield  {journal} {\bibinfo
  {journal} {Class. Quant. Grav.}\ }\textbf {\bibinfo {volume} {36}},\ \bibinfo
  {pages} {155007} (\bibinfo {year} {2019}{\natexlab{b}})},\ \Eprint
  {https://arxiv.org/abs/1811.00814} {arXiv:1811.00814 [gr-qc]} \BibitemShut
  {NoStop}%
\bibitem [{\citenamefont {Eichhorn}\ \emph
  {et~al.}(2020{\natexlab{a}})\citenamefont {Eichhorn}, \citenamefont {Lumma},
  \citenamefont {Pereira},\ and\ \citenamefont {Sikandar}}]{Eichhorn:2019hsa}%
  \BibitemOpen
  \bibfield  {author} {\bibinfo {author} {\bibfnamefont {A.}~\bibnamefont
  {Eichhorn}}, \bibinfo {author} {\bibfnamefont {J.}~\bibnamefont {Lumma}},
  \bibinfo {author} {\bibfnamefont {A.~D.}\ \bibnamefont {Pereira}},\ and\
  \bibinfo {author} {\bibfnamefont {A.}~\bibnamefont {Sikandar}},\ }\bibfield
  {title} {\emph {\bibinfo {title} {{Universal critical behavior in tensor
  models for four-dimensional quantum gravity}}},\ }\href
  {https://doi.org/10.1007/JHEP02(2020)110} {\bibfield  {journal} {\bibinfo
  {journal} {JHEP}\ }\textbf {\bibinfo {volume} {02}},\ \bibinfo {pages}
  {110}},\ \Eprint {https://arxiv.org/abs/1912.05314} {arXiv:1912.05314
  [gr-qc]} \BibitemShut {NoStop}%
\bibitem [{\citenamefont {Eichhorn}\ \emph
  {et~al.}(2020{\natexlab{b}})\citenamefont {Eichhorn}, \citenamefont {Lumma},
  \citenamefont {Pereira},\ and\ \citenamefont {Sikandar}}]{Eichhorn:2020gq}%
  \BibitemOpen
  \bibfield  {author} {\bibinfo {author} {\bibfnamefont {A.}~\bibnamefont
  {Eichhorn}}, \bibinfo {author} {\bibfnamefont {J.}~\bibnamefont {Lumma}},
  \bibinfo {author} {\bibfnamefont {A.~D.}\ \bibnamefont {Pereira}},\ and\
  \bibinfo {author} {\bibfnamefont {A.}~\bibnamefont {Sikandar}},\ }\bibfield
  {title} {\emph {\bibinfo {title} {{Universal critical behavior in tensor
  models for four-dimensional quantum gravity}}},\ }\href
  {https://doi.org/10.1007/JHEP02(2020)110} {\bibfield  {journal} {\bibinfo
  {journal} {JHEP}\ }\textbf {\bibinfo {volume} {02}},\ \bibinfo {pages}
  {110}},\ \Eprint {https://arxiv.org/abs/1912.05314} {1912.05314} \BibitemShut
  {NoStop}%
\bibitem [{\citenamefont {Lahoche}\ and\ \citenamefont
  {Ousmane~Samary}(2019)}]{Lahoche:2019wm}%
  \BibitemOpen
  \bibfield  {author} {\bibinfo {author} {\bibfnamefont {V.}~\bibnamefont
  {Lahoche}}\ and\ \bibinfo {author} {\bibfnamefont {D.}~\bibnamefont
  {Ousmane~Samary}},\ }\bibfield  {title} {\emph {\bibinfo {title} {{Large-$d$
  behavior of the Feynman amplitudes for a just-renormalizable tensorial group
  field theory}}},\ }\href@noop {} {\bibfield  {journal} {\bibinfo  {journal}
  {arXiv}\ ,\ \bibinfo {pages} {arXiv:1911.08601}} (\bibinfo {year} {2019})},\
  \Eprint {https://arxiv.org/abs/1911.08601} {1911.08601} \BibitemShut
  {NoStop}%
\bibitem [{\citenamefont {Lahoche}\ \emph {et~al.}(2020)\citenamefont
  {Lahoche}, \citenamefont {Ousmane~Samary},\ and\ \citenamefont
  {Pereira}}]{Lahoche:2020bi}%
  \BibitemOpen
  \bibfield  {author} {\bibinfo {author} {\bibfnamefont {V.}~\bibnamefont
  {Lahoche}}, \bibinfo {author} {\bibfnamefont {D.}~\bibnamefont
  {Ousmane~Samary}},\ and\ \bibinfo {author} {\bibfnamefont {A.~D.}\
  \bibnamefont {Pereira}},\ }\bibfield  {title} {\emph {\bibinfo {title}
  {{Renormalization group flow of coupled tensorial group field theories:
  Towards the Ising model on random lattices}}},\ }\href
  {https://doi.org/10.1093/ptep/ptv049} {\bibfield  {journal} {\bibinfo
  {journal} {Phys. Rev. D}\ }\textbf {\bibinfo {volume} {101}},\ \bibinfo
  {pages} {064014} (\bibinfo {year} {2020})},\ \Eprint
  {https://arxiv.org/abs/1911.05173} {1911.05173} \BibitemShut {NoStop}%
\bibitem [{\citenamefont {Lahoche}\ and\ \citenamefont
  {Ousmane~Samary}(2020{\natexlab{a}})}]{Lahoche:2020df}%
  \BibitemOpen
  \bibfield  {author} {\bibinfo {author} {\bibfnamefont {V.}~\bibnamefont
  {Lahoche}}\ and\ \bibinfo {author} {\bibfnamefont {D.}~\bibnamefont
  {Ousmane~Samary}},\ }\bibfield  {title} {\emph {\bibinfo {title} {{Revisited
  functional renormalization group approach for random matrices in the large-N
  limit}}},\ }\href {https://doi.org/10.1103/PhysRevD.101.106015} {\bibfield
  {journal} {\bibinfo  {journal} {Phys. Rev. D}\ }\textbf {\bibinfo {volume}
  {101}},\ \bibinfo {pages} {106015} (\bibinfo {year} {2020}{\natexlab{a}})},\
  \Eprint {https://arxiv.org/abs/1909.03327} {1909.03327} \BibitemShut
  {NoStop}%
\bibitem [{\citenamefont {Lahoche}\ and\ \citenamefont
  {Ousmane~Samary}(2020{\natexlab{b}})}]{Lahoche:2020ib}%
  \BibitemOpen
  \bibfield  {author} {\bibinfo {author} {\bibfnamefont {V.}~\bibnamefont
  {Lahoche}}\ and\ \bibinfo {author} {\bibfnamefont {D.}~\bibnamefont
  {Ousmane~Samary}},\ }\bibfield  {title} {\emph {\bibinfo {title}
  {{Pedagogical comments about nonperturbative Ward-constrained melonic
  renormalization group flow}}},\ }\href
  {https://doi.org/10.1103/PhysRevD.101.024001} {\bibfield  {journal} {\bibinfo
   {journal} {Phys. Rev. D}\ }\textbf {\bibinfo {volume} {101}},\ \bibinfo
  {pages} {024001} (\bibinfo {year} {2020}{\natexlab{b}})},\ \Eprint
  {https://arxiv.org/abs/2001.00934} {2001.00934} \BibitemShut {NoStop}%
\bibitem [{\citenamefont {Bonzom}\ \emph {et~al.}(2015)\citenamefont {Bonzom},
  \citenamefont {Delepouve},\ and\ \citenamefont {Rivasseau}}]{Bonzom:2015axa}%
  \BibitemOpen
  \bibfield  {author} {\bibinfo {author} {\bibfnamefont {V.}~\bibnamefont
  {Bonzom}}, \bibinfo {author} {\bibfnamefont {T.}~\bibnamefont {Delepouve}},\
  and\ \bibinfo {author} {\bibfnamefont {V.}~\bibnamefont {Rivasseau}},\
  }\bibfield  {title} {\emph {\bibinfo {title} {{Enhancing non-melonic
  triangulations: A tensor model mixing melonic and planar maps}}},\ }\href
  {https://doi.org/10.1016/j.nuclphysb.2015.04.004} {\bibfield  {journal}
  {\bibinfo  {journal} {Nucl. Phys. B}\ }\textbf {\bibinfo {volume} {895}},\
  \bibinfo {pages} {161--191} (\bibinfo {year} {2015})},\ \Eprint
  {https://arxiv.org/abs/1502.01365} {arXiv:1502.01365 [math-ph]} \BibitemShut
  {NoStop}%
\bibitem [{\citenamefont {Ben~Geloun}\ and\ \citenamefont
  {Toriumi}(2018)}]{Ben_Geloun_2018}%
  \BibitemOpen
  \bibfield  {author} {\bibinfo {author} {\bibfnamefont {J.}~\bibnamefont
  {Ben~Geloun}}\ and\ \bibinfo {author} {\bibfnamefont {R.}~\bibnamefont
  {Toriumi}},\ }\bibfield  {title} {\emph {\bibinfo {title} {Renormalizable
  enhanced tensor field theory: The quartic melonic case}},\ }\bibfield
  {journal} {\bibinfo  {journal} {Journal of Mathematical Physics}\ }\textbf
  {\bibinfo {volume} {59}},\ \href {https://doi.org/10.1063/1.5022438}
  {10.1063/1.5022438} (\bibinfo {year} {2018})\BibitemShut {NoStop}%
\bibitem [{\citenamefont {Ben~Geloun}\ and\ \citenamefont
  {Toriumi}(2024)}]{Geloun:2023oyd}%
  \BibitemOpen
  \bibfield  {author} {\bibinfo {author} {\bibfnamefont {J.}~\bibnamefont
  {Ben~Geloun}}\ and\ \bibinfo {author} {\bibfnamefont {R.}~\bibnamefont
  {Toriumi}},\ }\bibfield  {title} {\emph {\bibinfo {title} {{One-loop
  beta-functions of quartic enhanced tensor field theories}}},\ }\href
  {https://doi.org/10.1088/1751-8121/acfdde} {\bibfield  {journal} {\bibinfo
  {journal} {J. Phys. A}\ }\textbf {\bibinfo {volume} {57}},\ \bibinfo {pages}
  {015401} (\bibinfo {year} {2024})},\ \Eprint
  {https://arxiv.org/abs/2303.09829} {arXiv:2303.09829 [hep-th]} \BibitemShut
  {NoStop}%
\bibitem [{\citenamefont {Read}(1959)}]{Read1959TheEO}%
  \BibitemOpen
  \bibfield  {author} {\bibinfo {author} {\bibfnamefont {R.~C.}\ \bibnamefont
  {Read}},\ }\bibfield  {title} {\emph {\bibinfo {title} {The enumeration of
  locally restricted graphs (i)}},\ }\href
  {https://api.semanticscholar.org/CorpusID:120999014} {\bibfield  {journal}
  {\bibinfo  {journal} {Journal of The London Mathematical Society-second
  Series}\ ,\ \bibinfo {pages} {417--436}} (\bibinfo {year}
  {1959})}\BibitemShut {NoStop}%
\bibitem [{\citenamefont {de~Mello~Koch}\ and\ \citenamefont
  {Ramgoolam}(2012)}]{deMelloKoch:2011uq}%
  \BibitemOpen
  \bibfield  {author} {\bibinfo {author} {\bibfnamefont {R.}~\bibnamefont
  {de~Mello~Koch}}\ and\ \bibinfo {author} {\bibfnamefont {S.}~\bibnamefont
  {Ramgoolam}},\ }\bibfield  {title} {\emph {\bibinfo {title} {{Strings from
  Feynman Graph counting : without large N}}},\ }\href
  {https://doi.org/10.1103/PhysRevD.85.026007} {\bibfield  {journal} {\bibinfo
  {journal} {Phys. Rev. D}\ }\textbf {\bibinfo {volume} {85}},\ \bibinfo
  {pages} {026007} (\bibinfo {year} {2012})},\ \Eprint
  {https://arxiv.org/abs/1110.4858} {arXiv:1110.4858 [hep-th]} \BibitemShut
  {NoStop}%
\bibitem [{\citenamefont {Ben~Geloun}\ and\ \citenamefont
  {Ramgoolam}(2014)}]{BenGeloun:2013lim}%
  \BibitemOpen
  \bibfield  {author} {\bibinfo {author} {\bibfnamefont {J.}~\bibnamefont
  {Ben~Geloun}}\ and\ \bibinfo {author} {\bibfnamefont {S.}~\bibnamefont
  {Ramgoolam}},\ }\bibfield  {title} {\emph {\bibinfo {title} {{Counting tensor
  model observables and branched covers of the 2-sphere}}},\ }\href
  {https://doi.org/10.4171/aihpd/4} {\bibfield  {journal} {\bibinfo  {journal}
  {Ann. Inst. H. Poincare D Comb. Phys. Interact.}\ }\textbf {\bibinfo {volume}
  {1}},\ \bibinfo {pages} {77--138} (\bibinfo {year} {2014})},\ \Eprint
  {https://arxiv.org/abs/1307.6490} {arXiv:1307.6490 [hep-th]} \BibitemShut
  {NoStop}%
\bibitem [{\citenamefont {Avohou}\ \emph {et~al.}(2020)\citenamefont {Avohou},
  \citenamefont {Ben~Geloun},\ and\ \citenamefont {Dub}}]{Avohou_2020}%
  \BibitemOpen
  \bibfield  {author} {\bibinfo {author} {\bibfnamefont {R.~C.}\ \bibnamefont
  {Avohou}}, \bibinfo {author} {\bibfnamefont {J.}~\bibnamefont {Ben~Geloun}},\
  and\ \bibinfo {author} {\bibfnamefont {N.}~\bibnamefont {Dub}},\ }\bibfield
  {title} {\emph {\bibinfo {title} {On the counting of $o(n)$ tensor
  invariants}},\ }\href {https://doi.org/10.4310/atmp.2020.v24.n4.a1}
  {\bibfield  {journal} {\bibinfo  {journal} {Advances in Theoretical and
  Mathematical Physics}\ }\textbf {\bibinfo {volume} {24}},\ \bibinfo {pages}
  {821–878} (\bibinfo {year} {2020})}\BibitemShut {NoStop}%
\bibitem [{\citenamefont {Avohou}\ \emph {et~al.}(2024)\citenamefont {Avohou},
  \citenamefont {Ben~Geloun},\ and\ \citenamefont {Toriumi}}]{Avohou:2024agh}%
  \BibitemOpen
  \bibfield  {author} {\bibinfo {author} {\bibfnamefont {R.~C.}\ \bibnamefont
  {Avohou}}, \bibinfo {author} {\bibfnamefont {J.}~\bibnamefont {Ben~Geloun}},\
  and\ \bibinfo {author} {\bibfnamefont {R.}~\bibnamefont {Toriumi}},\
  }\bibfield  {title} {\emph {\bibinfo {title} {{Counting $U(N)^{\otimes
  r}\otimes O(N)^{\otimes q}$ invariants and tensor model observables}}},\
  }\href {https://doi.org/10.1140/epjc/s10052-024-13091-z} {\bibfield
  {journal} {\bibinfo  {journal} {Eur. Phys. J. C}\ }\textbf {\bibinfo {volume}
  {84}},\ \bibinfo {pages} {839} (\bibinfo {year} {2024})},\ \Eprint
  {https://arxiv.org/abs/2404.16404} {arXiv:2404.16404 [hep-th]} \BibitemShut
  {NoStop}%
\bibitem [{\citenamefont {Ben~Geloun}\ and\ \citenamefont
  {Ramgoolam}(2017)}]{BenGeloun:2017vwn}%
  \BibitemOpen
  \bibfield  {author} {\bibinfo {author} {\bibfnamefont {J.}~\bibnamefont
  {Ben~Geloun}}\ and\ \bibinfo {author} {\bibfnamefont {S.}~\bibnamefont
  {Ramgoolam}},\ }\bibfield  {title} {\emph {\bibinfo {title} {{Tensor Models,
  Kronecker coefficients and Permutation Centralizer Algebras}}},\ }\href
  {https://doi.org/10.1007/JHEP11(2017)092} {\bibfield  {journal} {\bibinfo
  {journal} {JHEP}\ }\textbf {\bibinfo {volume} {11}},\ \bibinfo {pages}
  {092}},\ \Eprint {https://arxiv.org/abs/1708.03524} {arXiv:1708.03524
  [hep-th]} \BibitemShut {NoStop}%
\bibitem [{\citenamefont {Ben~Geloun}\ and\ \citenamefont
  {Ramgoolam}(2023)}]{BenGeloun:2020yau}%
  \BibitemOpen
  \bibfield  {author} {\bibinfo {author} {\bibfnamefont {J.}~\bibnamefont
  {Ben~Geloun}}\ and\ \bibinfo {author} {\bibfnamefont {S.}~\bibnamefont
  {Ramgoolam}},\ }\bibfield  {title} {\emph {\bibinfo {title} {{Quantum
  mechanics of bipartite ribbon graphs: Integrality, Lattices and Kronecker
  coefficients}}},\ }\href {https://doi.org/10.5802/alco.254} {\bibfield
  {journal} {\bibinfo  {journal} {Alg. Comb.}\ }\textbf {\bibinfo {volume}
  {6}},\ \bibinfo {pages} {547--594} (\bibinfo {year} {2023})},\ \Eprint
  {https://arxiv.org/abs/2010.04054} {arXiv:2010.04054 [hep-th]} \BibitemShut
  {NoStop}%
\bibitem [{\citenamefont {Ben~Geloun}(2020)}]{BenGeloun:2020lfe}%
  \BibitemOpen
  \bibfield  {author} {\bibinfo {author} {\bibfnamefont {J.}~\bibnamefont
  {Ben~Geloun}},\ }\bibfield  {title} {\emph {\bibinfo {title} {{On the
  counting tensor model observables as $U(N)$ and $O(N)$ classical
  invariants}}},\ }\href {https://doi.org/10.22323/1.376.0175} {\bibfield
  {journal} {\bibinfo  {journal} {PoS}\ }\textbf {\bibinfo {volume}
  {CORFU2019}},\ \bibinfo {pages} {175} (\bibinfo {year} {2020})},\ \Eprint
  {https://arxiv.org/abs/2005.01773} {arXiv:2005.01773 [hep-th]} \BibitemShut
  {NoStop}%
\bibitem [{\citenamefont {Ben~Geloun}\ and\ \citenamefont
  {Ramgoolam}(2022)}]{BenGeloun:2021cuj}%
  \BibitemOpen
  \bibfield  {author} {\bibinfo {author} {\bibfnamefont {J.}~\bibnamefont
  {Ben~Geloun}}\ and\ \bibinfo {author} {\bibfnamefont {S.}~\bibnamefont
  {Ramgoolam}},\ }\bibfield  {title} {\emph {\bibinfo {title} {{All-orders
  asymptotics of tensor model observables from symmetries of restricted
  partitions}}},\ }\href {https://doi.org/10.1088/1751-8121/ac9b3b} {\bibfield
  {journal} {\bibinfo  {journal} {J. Phys. A}\ }\textbf {\bibinfo {volume}
  {55}},\ \bibinfo {pages} {435203} (\bibinfo {year} {2022})},\ \Eprint
  {https://arxiv.org/abs/2106.01470} {arXiv:2106.01470 [hep-th]} \BibitemShut
  {NoStop}%
\bibitem [{\citenamefont {Dijkgraaf}\ and\ \citenamefont {Witten}(1990)}]{DW}%
  \BibitemOpen
  \bibfield  {author} {\bibinfo {author} {\bibfnamefont {R.}~\bibnamefont
  {Dijkgraaf}}\ and\ \bibinfo {author} {\bibfnamefont {E.}~\bibnamefont
  {Witten}},\ }\bibfield  {title} {\emph {\bibinfo {title} {Topological gauge
  theories and group cohomology}},\ }\href {https://doi.org/10.1007/BF02096988}
  {\bibfield  {journal} {\bibinfo  {journal} {Communications in Mathematical
  Physics}\ }\textbf {\bibinfo {volume} {129}},\ \bibinfo {pages} {393--429}
  (\bibinfo {year} {1990})}\BibitemShut {NoStop}%
\bibitem [{\citenamefont {Barcelo}\ and\ \citenamefont
  {Ram}(1997)}]{BarceloRam1997}%
  \BibitemOpen
  \bibfield  {author} {\bibinfo {author} {\bibfnamefont {H.}~\bibnamefont
  {Barcelo}}\ and\ \bibinfo {author} {\bibfnamefont {A.}~\bibnamefont {Ram}},\
  }\bibfield  {title} {\emph {\bibinfo {title} {Combinatorial representation
  theory}},\ }\href {https://arxiv.org/abs/math/9707221} {\bibfield  {journal}
  {\bibinfo  {journal} {arXiv Mathematics e-prints}\ } (\bibinfo {year}
  {1997})},\ \Eprint {https://arxiv.org/abs/math/9707221} {arXiv:math/9707221
  [math.RT]} \BibitemShut {NoStop}%
\bibitem [{\citenamefont {Murnaghan}(1938)}]{Murnaghan1938TheAO}%
  \BibitemOpen
  \bibfield  {author} {\bibinfo {author} {\bibfnamefont {F.~D.}\ \bibnamefont
  {Murnaghan}},\ }\bibfield  {title} {\emph {\bibinfo {title} {The analysis of
  the kronecker product of irreducible representations of the symmetric
  group}},\ }\href@noop {} {\bibfield  {journal} {\bibinfo  {journal} {American
  Journal of Mathematics}\ }\textbf {\bibinfo {volume} {60}},\ \bibinfo {pages}
  {761} (\bibinfo {year} {1938})}\BibitemShut {NoStop}%
\bibitem [{\citenamefont {Padellaro}\ \emph {et~al.}(2025)\citenamefont
  {Padellaro}, \citenamefont {Ramgoolam},\ and\ \citenamefont
  {Seong}}]{Padellaro:2025jtd}%
  \BibitemOpen
  \bibfield  {author} {\bibinfo {author} {\bibfnamefont {A.}~\bibnamefont
  {Padellaro}}, \bibinfo {author} {\bibfnamefont {S.}~\bibnamefont
  {Ramgoolam}},\ and\ \bibinfo {author} {\bibfnamefont {R.-K.}\ \bibnamefont
  {Seong}},\ }\href {https://arxiv.org/abs/2503.02543} {\bibinfo {title} {Row
  and column detection complexities of character tables}} (\bibinfo {year}
  {2025}),\ \Eprint {https://arxiv.org/abs/2503.02543} {arXiv:2503.02543
  [hep-th]} \BibitemShut {NoStop}%
\bibitem [{\citenamefont {de~Mello~Koch}\ \emph {et~al.}(2022)\citenamefont
  {de~Mello~Koch}, \citenamefont {He}, \citenamefont {Kemp},\ and\
  \citenamefont {Ramgoolam}}]{dMelloKoch:2021}%
  \BibitemOpen
  \bibfield  {author} {\bibinfo {author} {\bibfnamefont {R.}~\bibnamefont
  {de~Mello~Koch}}, \bibinfo {author} {\bibfnamefont {Y.-H.}\ \bibnamefont
  {He}}, \bibinfo {author} {\bibfnamefont {G.}~\bibnamefont {Kemp}},\ and\
  \bibinfo {author} {\bibfnamefont {S.}~\bibnamefont {Ramgoolam}},\ }\bibfield
  {title} {\emph {\bibinfo {title} {{Integrality, duality and finiteness in
  combinatoric topological strings}}},\ }\href
  {https://doi.org/10.1007/JHEP01(2022)071} {\bibfield  {journal} {\bibinfo
  {journal} {JHEP}\ }\textbf {\bibinfo {volume} {01}},\ \bibinfo {pages}
  {071}},\ \Eprint {https://arxiv.org/abs/2106.05598} {arXiv:2106.05598
  [hep-th]} \BibitemShut {NoStop}%
\bibitem [{\citenamefont {Ramgoolam}\ and\ \citenamefont
  {Sharpe}(2022)}]{Ramgoolam:2022xfk}%
  \BibitemOpen
  \bibfield  {author} {\bibinfo {author} {\bibfnamefont {S.}~\bibnamefont
  {Ramgoolam}}\ and\ \bibinfo {author} {\bibfnamefont {E.}~\bibnamefont
  {Sharpe}},\ }\bibfield  {title} {\emph {\bibinfo {title} {{Combinatoric
  topological string theories and group theory algorithms}}},\ }\href
  {https://doi.org/10.1007/JHEP10(2022)147} {\bibfield  {journal} {\bibinfo
  {journal} {JHEP}\ }\textbf {\bibinfo {volume} {10}},\ \bibinfo {pages}
  {147}},\ \Eprint {https://arxiv.org/abs/2204.02266} {arXiv:2204.02266
  [hep-th]} \BibitemShut {NoStop}%
\bibitem [{\citenamefont {Padellaro}\ \emph {et~al.}(2024)\citenamefont
  {Padellaro}, \citenamefont {Radhakrishnan},\ and\ \citenamefont
  {Ramgoolam}}]{Padellaro:2023oqj}%
  \BibitemOpen
  \bibfield  {author} {\bibinfo {author} {\bibfnamefont {A.}~\bibnamefont
  {Padellaro}}, \bibinfo {author} {\bibfnamefont {R.}~\bibnamefont
  {Radhakrishnan}},\ and\ \bibinfo {author} {\bibfnamefont {S.}~\bibnamefont
  {Ramgoolam}},\ }\bibfield  {title} {\emph {\bibinfo {title}
  {{Row\textendash{}column duality and combinatorial topological strings}}},\
  }\href {https://doi.org/10.1088/1751-8121/ad1d24} {\bibfield  {journal}
  {\bibinfo  {journal} {J. Phys. A}\ }\textbf {\bibinfo {volume} {57}},\
  \bibinfo {pages} {065202} (\bibinfo {year} {2024})},\ \Eprint
  {https://arxiv.org/abs/2304.10217} {arXiv:2304.10217 [hep-th]} \BibitemShut
  {NoStop}%
\bibitem [{\citenamefont {Lang}(2002)}]{Lang2002}%
  \BibitemOpen
  \bibfield  {author} {\bibinfo {author} {\bibfnamefont {S.}~\bibnamefont
  {Lang}},\ }\href@noop {} {\emph {\bibinfo {title} {Algebra}}},\ \bibinfo
  {edition} {3rd}\ ed.\ (\bibinfo  {publisher} {Springer},\ \bibinfo {address}
  {New York},\ \bibinfo {year} {2002})\BibitemShut {NoStop}%
\bibitem [{\citenamefont {Oriti}\ \emph {et~al.}(2017)\citenamefont {Oriti},
  \citenamefont {Sindoni},\ and\ \citenamefont {Wilson-Ewing}}]{Oriti:2016ueo}%
  \BibitemOpen
  \bibfield  {author} {\bibinfo {author} {\bibfnamefont {D.}~\bibnamefont
  {Oriti}}, \bibinfo {author} {\bibfnamefont {L.}~\bibnamefont {Sindoni}},\
  and\ \bibinfo {author} {\bibfnamefont {E.}~\bibnamefont {Wilson-Ewing}},\
  }\bibfield  {title} {\emph {\bibinfo {title} {{Bouncing cosmologies from
  quantum gravity condensates}}},\ }\href
  {https://doi.org/10.1088/1361-6382/aa549a} {\bibfield  {journal} {\bibinfo
  {journal} {Class. Quant. Grav.}\ }\textbf {\bibinfo {volume} {34}},\ \bibinfo
  {pages} {04LT01} (\bibinfo {year} {2017})},\ \Eprint
  {https://arxiv.org/abs/1602.08271} {arXiv:1602.08271 [gr-qc]} \BibitemShut
  {NoStop}%
\bibitem [{\citenamefont {Grosse}\ and\ \citenamefont
  {Wulkenhaar}(2005)}]{Grosse_2005}%
  \BibitemOpen
  \bibfield  {author} {\bibinfo {author} {\bibfnamefont {H.}~\bibnamefont
  {Grosse}}\ and\ \bibinfo {author} {\bibfnamefont {R.}~\bibnamefont
  {Wulkenhaar}},\ }\bibfield  {title} {\emph {\bibinfo {title} {Renormalisation
  of ?4-theory on noncommutative ?4 in the matrix base}},\ }\href
  {https://doi.org/10.1007/s00220-004-1285-2} {\bibfield  {journal} {\bibinfo
  {journal} {Communications in Mathematical Physics}\ }\textbf {\bibinfo
  {volume} {256}},\ \bibinfo {pages} {305–374} (\bibinfo {year}
  {2005})}\BibitemShut {NoStop}%
\bibitem [{\citenamefont {Ben~Geloun}\ and\ \citenamefont
  {Koslowski}(2016)}]{BenGeloun:2016tmc}%
  \BibitemOpen
  \bibfield  {author} {\bibinfo {author} {\bibfnamefont {J.}~\bibnamefont
  {Ben~Geloun}}\ and\ \bibinfo {author} {\bibfnamefont {T.~A.}\ \bibnamefont
  {Koslowski}},\ }\href {https://arxiv.org/abs/1606.04044} {\bibinfo {title}
  {Nontrivial uv behavior of rank-4 tensor field models for quantum gravity}}
  (\bibinfo {year} {2016}),\ \Eprint {https://arxiv.org/abs/1606.04044}
  {arXiv:1606.04044 [gr-qc]} \BibitemShut {NoStop}%
\end{thebibliography}%
\end{document}